\documentclass[12pt]{article}
\usepackage{amsmath,amssymb,amsfonts,amscd,amsxtra,amsthm}
\usepackage{latexsym}
\usepackage{xcolor}
\usepackage{graphicx}
\usepackage{subfig}
\usepackage{float}
\usepackage{overpic}
\usepackage{algorithmic}
\usepackage[ruled,linesnumbered,noend]{algorithm2e}
\usepackage{booktabs}
\usepackage{appendix}
\usepackage[top=1in, bottom=1in, left=1in, right=1in]{geometry}
\renewcommand{\theequation}{\thesection.\arabic{equation}}

\newtheorem{thm}{Theorem}[section]

\newtheorem{lem}[thm]{Lemma}
\newtheorem{prop}[thm]{Proposition}
\newtheorem{defn}{Definition}

\newtheorem{rmk}[thm]{Remark}

\renewcommand{\Im}{{\mbox{Im}}}

\newcommand{\Real}{\mathbb R}

\renewcommand\appendix{\par
  \setcounter{section}{0}
  \setcounter{subsection}{0}
  \setcounter{figure}{0}
  \setcounter{table}{0}
  \renewcommand\thesection{Appendix \Alph{section}}
  \renewcommand\theequation{\Alph{section}.\arabic{equation}}
  \renewcommand\thefigure{\Alph{section}.\arabic{figure}}
  \renewcommand\thetable{\Alph{section}.\arabic{table}}
  \renewcommand\thethm{\Alph{section}.\arabic{thm}}
}

\numberwithin{equation}{section}

\date{}

\title{SCAN-MUSIC: An Efficient Super-resolution Algorithm for Single Snapshot Wide-band Line Spectral Estimation }
 \author{
Zetao Fei
\thanks{Department of Mathematics, 
HKUST,  Clear Water Bay, Kowloon, Hong Kong (zfei@connect.ust.hk).}
and Hai Zhang
\thanks{Department of Mathematics, 
HKUST,  Clear Water Bay, Kowloon, Hong Kong (haizhang@ust.hk).}
}

\begin{document}

\maketitle

\begin{abstract}
We propose an efficient algorithm for reconstructing one-dimensional wide-band line spectra from their Fourier data in a bounded interval $[-\Omega,\Omega]$. While traditional subspace methods such as MUSIC achieve super-resolution for closely separated line spectra, their computational cost is high, particularly for wide-band line spectra. To address this issue, we proposed a scalable algorithm termed SCAN-MUSIC that scans the spectral domain using a fixed Gaussian window and then reconstructs the line spectra falling into the window at each time. For line spectra with cluster structure, we further refine the proposed algorithm using the annihilating filter technique. 
Both algorithms can significantly reduce the computational complexity of the standard MUSIC algorithm with a moderate loss of resolution. Moreover, in terms of speed, their performance is comparable to the state-of-the-art algorithms, while being more reliable for reconstructing line spectra with cluster structure. 
The algorithms are supplemented with theoretical analyses of error estimates, sampling complexity, computational complexity, and computational limit. 

\end{abstract}

\section{Introduction}{\label{sec: introduction}}
Wide-band signal processing has gained significant attention in various fields including audio and speech processing \cite{ambikairajah2001wideband}, wireless communications \cite{shen2006ultra}, array processing \cite{krolik1990focused}\cite{liu2010wideband}, medical imaging\cite{do2018cardiac}, and so on. Various strategies are studied to achieve wide-band signal reconstruction, see e.g. \cite{wang1985coherent}\cite{di2001waves}\cite{yoon2006tops}\cite{chen2002maximum}\cite{valaee1995wideband}\cite{mishali2010theory}\cite{zheng2019radar}\cite{cevher2008sparse} and references therein. In this paper, we consider the wide-band line spectra estimation (LSE) problem. We assume the line spectra are spread over some interval $[-R, R]$ with $R \gg 1$. Mathematically, the problem can be formulated as follows.
Let $\nu=\sum_{j=1}^n a_j\delta_{y_j}$ be a discrete measure, where $y_j\in [-R,R]$ represents the support of the line spectra and $a_j$ the corresponding amplitudes. The noisy measurement is given by the sampled band-limited Fourier data:
\begin{align}{\label{measurement}}
Y(\omega_k) = \sum_{j=1}^{n} a_je^{iy_j\omega_k}+W(\omega_k),\ \omega_k\in[-\Omega,\Omega]
\end{align}
where $\omega_k = \frac{k}{K}\Omega$, $k = -K,\cdots,K$, are the uniform sampling points over the interval $[-\Omega,\Omega]$, $\Omega$ is the cutoff frequency and $W(\omega_k)$'s are the noise terms. The sampling step size is given by $h = \frac{\Omega}{K}$. Assume that $|W(\omega_k)|<\sigma$ for all $k$ with $\sigma$ denoting the noise level. Notice that the Rayleigh limit of the system given by (\ref{measurement}) is $\frac{\pi}{\Omega}$. We define the density of line spectra as 
\begin{align}
    \rho:= \frac{n}{2R}\cdot\frac{\pi}{\Omega},
\end{align}
which gives the average number of spectra per unit Rayleigh limit. We do not assume any specific noise pattern throughout the paper.\\

Different regimes have been studied to address the problem of line spectral estimation. When line spectra are extremely sparsely positioned over the real axis, the Sparse Fourier Transform (SFT) offers an efficient method for computing the Discrete Fourier transform using only a subset of the sampling data. Estimating the spectral position plays an important role in the algorithms. During its early development, SFT achieves computational complexity $n\log^{\mathcal{O}(1)}K$, where $n$ is the sparsity of the line spectrum and $K$ is the size of the signal, see e.g. \cite{gilbert2005improved}. All the algorithms are random algorithms with failing probability. Over the recent decade, SFT has been extensively researched. The efficient, stable, and accessible random algorithms can be found in various papers e.g. \cite{indyk2014nearly}\cite{price2015robust}\cite{chen2016fourier}. A survey on the main techniques and more details can be found in \cite{gilbert2014recent}. For deterministic algorithms of SFT, the state-of-art computational complexity for line spectral with no structure is of the order $n^2\log^{\mathcal{O}(1)}K$, see e.g. \cite{akavia2010deterministic}\cite{iwen2013improved}. There is another group of deterministic SFT algorithms called the Prony-based algorithm that requires $\mathcal{O}(n^3)$ computational complexity, see \cite{heider2013sparse}\cite{potts2016efficient}. The recent progress on the SFT algorithms can be found in e.g. \cite{bittens2019deterministic}\cite{choi2021high}.\\

When the separation distance between line spectra is multiple of the Rayleigh limit, the sparsity-exploiting algorithms are applicable. In \cite{candes2014towards}, \cite{candes2013super}, the TV minimization algorithm is proposed and the authors derive the reconstruction error bound given the separation distance is above 4 Rayleigh limits. The separation condition is further improved under certain conditions in \cite{10.1093/imaiai/iaw005}. We refer other sparsity-exploiting algorithms, such as Atomic Norm minimization, SBL, etc., to \cite{tang2014near}, \cite{tang2013compressed}, \cite{chi2020harnessing}, \cite{duval2015exact}, \cite{AZAIS2015177}, \cite{hansen2014sparse}, \cite{1315936} and the references therein. We also notice that in the recent work \cite{hansen2018superfast}, based on the Bayesian view, the authors develop an efficient algorithm by exploiting superfast Toeplitz inversion. \\

When the minimum separation distance approaches or even falls below the Rayleigh limit, classical subspace methods demonstrate excellent performance. Subspace methods, such as MUSIC \cite{1143830}, ESPRIT \cite{32276}, and Matrix Pencil \cite{hua1990matrix}, leverage eigen-decomposition to achieve super-resolution capabilities. The efficiency of such methods, however, may become a problem when the number of line spectra is large. The practical application of the MUSIC algorithm is hindered by its high computational complexity, typically $\mathcal{O}(K^3)$ where $K$ is the sampling complexity. Though the Matrix Pencil and ESPRIT are more efficient, they are very sensitive to the prior information on the source number. One of the goals of this paper is to speed up the standard MUSIC algorithm by reducing its computational complexity to $\mathcal{O}(K^2\log K)$. \\

Recent research has delved into studying the stability of subspace methods in various scenarios, as evident in works like \cite{LI2021118} and \cite{li2019super}. Interested readers can refer to \cite{batenkov2019super}, \cite{donoho1992superresolution}, \cite{demanet2015recoverability}, \cite{Moitra:2015:SEF:2746539.2746561}, \cite{9410626}, and \cite{LIU2022402} for detailed discussions on computational resolution limits.
We highlight that the theoretical results depend on crucial factors such as the signal-to-noise ratio (SNR) and the number of line spectra, which are discussed in Section \ref{subsec: tradeoff}.\\

\subsection{Our Contribution}\label{subsec: contribution}

The contribution of this work is to propose a scalable efficient algorithm for the wide-band Line Spectral Estimation (LSE) problem. This is achieved in the regime when the number of spectra is large and the average distance between neighboring spectra is close to or comparable to the Rayleigh limit. The proposed algorithm is based on the strategy of divide-and-conquer; it first divides the whole problem into multiple subproblems by centralization and Gaussian windowing and then solves each subproblem using the MUSIC algorithm, and hence is termed SCAN-MUSIC. We applied the algorithm to the case when the spectra are randomly distributed with an average separation distance above the Rayleigh limit.  We show that SCAN-MUSIC has the computational complexity $\mathcal{O}(n^2\log n)$ and the optimal sampling complexity $\mathcal{O}(n)$ under a proper choice of parameters.\\

For the case when the line spectra are clustered with separation distances smaller than the Rayleigh limit within clusters. We refine SCAN-MUSIC by incorporating the technique of annihilating filters and term it SCAN-MUSIC(C). SCAN-MUSIC(C) has optimal sampling complexity $\mathcal{O}(n)$ and computational complexity $\mathcal{O}(Tn\log n)$ 
where $T$ is the number of clusters. \\

We supplement the proposed algorithms with theoretical analysis of error estimates, sampling complexity, computational complexity, and computational limit. Numerical experiments are also conducted to demonstrate the efficiency of the proposed algorithms. It is shown that the proposed algorithm is comparable in speed to the state-of-the-art algorithm proposed in \cite{hansen2018superfast}. Moreover, it has unique strength in reconstructing line spectra with cluster structure where the separation distance between spectra is below the Rayleigh limit within clusters. We want to point out that both proposed algorithms can achieve robust reconstruction that is close to the theoretical computational limit. Moreover, they can
be sped up by parallel computing techniques. Compared to the algorithm in \cite{hansen2018superfast}, they are less sensitive to the statistics of noise. \\

Notably, in a recent study \cite{liu2022measurement}, the authors also consider the LSE problem with cluster structure; however, they adopt a different approach based on a downsampling strategy and multipole expansion trick.

\subsection{Organization of the paper}
In Section \ref{sec: random spectra reconstruction}, we explain the main idea of SCAN-MUSIC and provide theoretical discussion on the computational limit, sampling and computational complexity, and error estimates. In Section \ref{sec: cluster spectra}, we consider the LSE problem for clustered line spectra. Details of the implementation of the above two algorithms are provided in Section \ref{sec: details}. In Section \ref{sec: numerical study}, the numerical results are presented. The paper is concluded by discussions in Section \ref{sec: discussion}.

\subsection{Assumptions and notations}
Throughout the paper, we assume that $K\ge n$ and that the sampling step size follows the Nyquist-Shannon criterion, i.e. $h\le\frac{\pi}{R}$. We point out that we do not assume any specific pattern of the noise pattern. \\

For ease of notation, we denote 
\begin{align}\label{fourier data function}
    f(\omega) = \sum_{j=1}^{n} a_je^{iy_j\omega},
\end{align}
and for a slight abuse of notation, we denote
\begin{align}
    Y(\omega) = f(\omega)+W(\omega).
\end{align}
We denote $m_{\min} = \min_{j}|a_j|$, and $d_{\min} = \min_{s\ne t} |y_s-y_t|$.
We denote the Fourier transform operator as $\mathcal{F}$, and its inverse as $\mathcal{F}^{-1}$. The Fourier transform used in this paper is defined as $\mathcal{F}g(x) = \frac{1}{2\pi}\int_\Real g(\omega)\cdot e^{-ix\omega} d\omega$ for some function $g$. Let $\chi_{A}$ be the characteristic function of some set $A$. We denote $\|\cdot\|_{TV}$ the total variation norm and $\|\cdot\|_{\infty}$ the $L^{\infty}$ norm. Let $\Phi(x) = \int_{-\infty}^x e^{-t^2}dt$. 
For readers' convenience, we list the parameters used in the introduction of algorithms in Section \ref{subsec: notations}.\\

Finally, we clarify the difference between the subsampling strategy and downsampling strategy that occur in this paper: subsampling strategy is the strategy enlarging the sampling step size; downsampling strategy is the strategy sampling in a subinterval in $[-\Omega,\Omega]$.

\section{Reconstruction of Random Line Spectra}\label{sec: random spectra reconstruction}
In this section, we consider the LSE problem for the case that the line spectra are randomly separated on the interval $[-R, R]$ with no specific structure. In Section \ref{subsec: windowing and centralization}, we present the main idea of centralization and Gaussian windowing. Based on these two techniques, the SCAN-MUSIC method is proposed, in Section \ref{subsec: scan-music}. The sampling complexity is discussed in Section \ref{subsec: sampling complexity}. In Section \ref{subsec: random-compute-complexity}, we discuss the computational complexity of SCAN-MUSIC. We consider the application of SCAN-MUSIC in the super-sparse regime in Section \ref{subsec: super-sparse regime}. Finally, the trade-off between computational complexity and resolution loss is discussed in Section \ref{subsec: tradeoff}.
\subsection{Centralization and Gaussian windowing}\label{subsec: windowing and centralization}
Before going to the theoretical analysis of centralization and Gaussian windowing, we first explain the idea behind the SCAN-MUSIC proposed in Section \ref{subsec: scan-music}.\\

In the mathematical model (\ref{measurement}), the point spread function of the system is given by 
\begin{align}
    p(x) := \mathcal{F}\chi_{[-\Omega,\Omega]} = \frac{1}{2\pi}\int_{\Real} \chi_{[-\Omega,\Omega]}e^{-ix\omega} d \omega= \frac{1}{\pi x}\sin{\Omega x}.
\end{align}
The decaying rate of $p(x)$ is of the order $\mathcal{O}(\frac{1}{|x|})$. It implies that there exists strong interference between signals generated by different line spectra especially when the line spectra are not well separated. To address this issue, we introduce Gaussian windowing for the measurement. We define the normalized Gaussian kernel with parameter $\lambda$ in the frequency domain as follows:
\begin{align}
    G_{\lambda}(\omega) =\sqrt{\frac{\lambda}{\pi}} \cdot e^{-\lambda \omega^2}, \quad \omega \in \Real.
\end{align}
The Fourier transform of $G_{\lambda}(\omega)$ is given by
\begin{align}
    \mathcal{F}G_{\lambda}(x) = \sqrt{\frac{\lambda}{\pi}}\int_{\Real} e^{-\lambda \omega^2}e^{-i\omega x} d\omega = e^{-\frac{x^2}{4\lambda}}.
\end{align}
We expect that after the windowing, the line spectra that are far away from the origin are ``quenched''. To reconstruct the line spectra near $\mu\in [-R,R]$, a centralization step is needed. We define the centralization operator $\mathcal{S}_{\mu}:\mathcal{C(\Real)}\rightarrow\mathcal{C(\Real)}$ by
\begin{align}
    \mathcal{S}_{\mu}g = g \cdot e^{-i\mu \omega},
\end{align}
where $\mu\in [-R,R]$.\\

\begin{figure}
\centering
    \includegraphics[width= 0.8\textwidth]{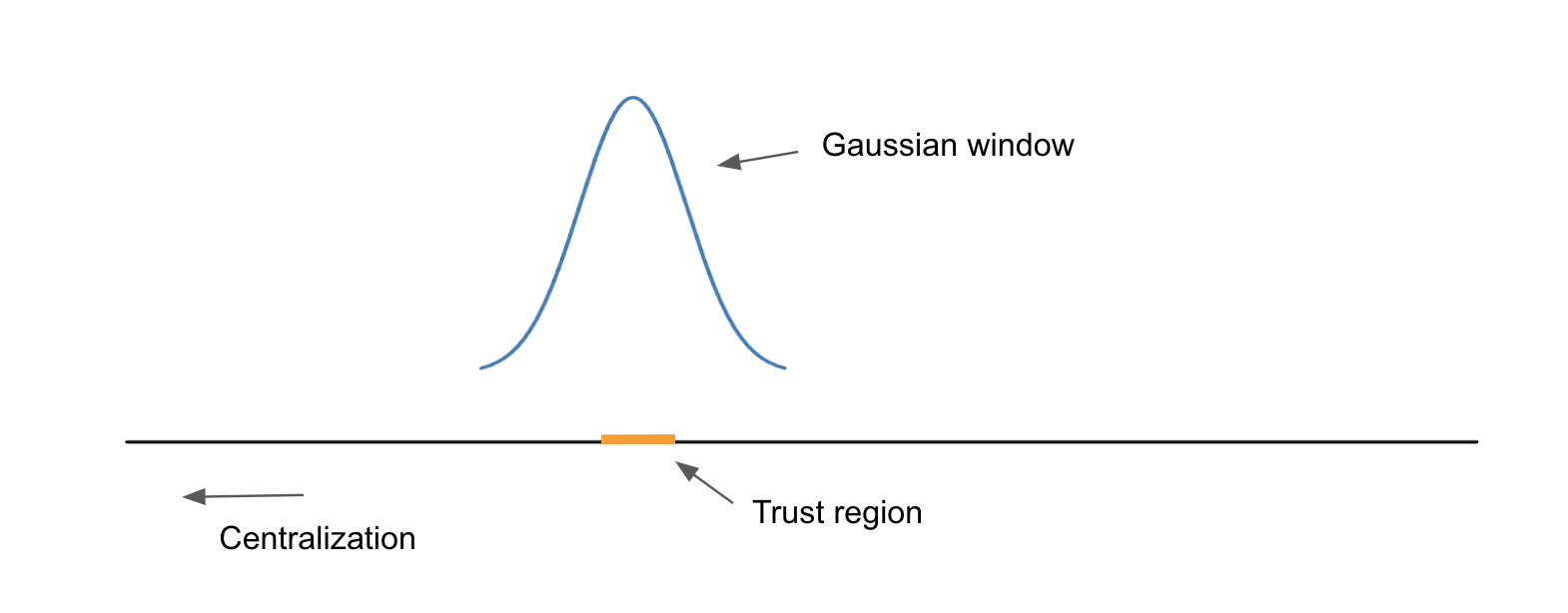}
\caption{Sketch of centralization and Gaussian windowing in the spatial domain.}
\label{Fig: windowing-centralization}
\end{figure}

Let $f$ be defined in (\ref{fourier data function}). By Lemma \ref{lem: gauss-conv}, the line spectra after centralization and Gaussian windowing can be characterized by the following equation:
\begin{align}\label{center-windowing}
    \left(\mathcal{S}_{\mu}(f\chi_{[-\Omega,\Omega]})\right) \ast G_{\lambda}
   & = \left(\sum_{j=1}^n a_{j}e^{i(y_{j}-\mu)\omega}\right) \ast G_{\lambda} + \mathcal{E}_{mod}(\omega)\notag\\
   & = \sum_{j=1}^{n} a_j e^{-\frac{(y_j-\mu)^2}{4\lambda}}e^{i(y_j-\mu)\omega} + \mathcal{E}_{mod}(\omega),
\end{align}
where $\mathcal{E}_{mod}$ is defined by 
\begin{align}\label{model-error}
    \mathcal{E}_{mod}(\omega) = \mathcal{S}_{\mu}\left(f\cdot\left(\chi_{[-\Omega,\Omega]}-1\right)\right) \ast G_{\lambda}(\omega).
\end{align}
We notice that the factor $e^{-\frac{(y_j-\mu)^2}{4\lambda}}$ in (\ref{center-windowing}) implies that the line spectra far away from the chosen center, $\mu$, are indeed ``quenched''. Therefore, the centralization and Gaussian windowing decouple the line spectra around the center and the ones that are sufficiently far away. Figure \ref{Fig: windowing-centralization} illustrates the procedure of centralization and Gaussian windowing.\\

We call $\mathcal{E}_{mod}$ in (\ref{model-error}) the model error, which is caused by the fact that only the band-limited data is available. 
The model error is estimated in the following theorem.
\begin{thm} For $0<\epsilon<1$ and $\omega\in [(-1+\epsilon)\Omega,(1-\epsilon)\Omega]$, 
\begin{align}
    |\mathcal{E}_{mod}(\omega)| \le \frac{\|\nu\|_{TV}}{\sqrt{\lambda}}\left( \Phi\left(-\sqrt{\lambda}\epsilon\Omega\right)+\Phi\left(\sqrt{\lambda}(-2\Omega+\epsilon\Omega)\right) \right).
\end{align}
\end{thm}
\begin{proof}
   By the definition of $\mathcal{E}_{mod}$,
   \begin{align}\label{E_mod}
       |\mathcal{E}_{mod}(\omega)|
       & = \sqrt{\frac{\lambda}{\pi}}\left|\sum_{j = 1}^n a_j \left( \int_{\Omega}^{+\infty} e^{i(y_j-\mu)\eta}\cdot e^{-\lambda(\omega-\eta)^2} d\eta + \int_{-\infty}^{-\Omega} e^{i(y_j-\mu)\eta}\cdot e^{-\lambda(\omega-\eta)^2} d\eta\right)\right|\notag \\
       & \le \sqrt{\frac{\lambda}{\pi}} \|\nu\|_{TV} \cdot \left( \int_{\Omega}^{+\infty} e^{-\lambda(\omega-\eta)^2} d\eta + \int_{-\infty}^{-\Omega} e^{-\lambda(\omega-\eta)^2} d\eta\right)\notag\\
       & = \frac{\|\nu\|_{TV}}{\sqrt{\pi}}\left( \Phi\left(\sqrt{\lambda}(-\Omega+\omega)\right)+\Phi\left(\sqrt{\lambda}(-\Omega-\omega)\right) \right).
   \end{align}
   It is clear that the right-hand side of (\ref{E_mod}) is an even function of $\omega$ and is increasing for $\omega\in [0,(1-\epsilon)\Omega]$. Therefore, we have
   \begin{align}
       |\mathcal{E}_{mod}(\omega)| \le \frac{\|\nu\|_{TV}}{\sqrt{\pi}}\left( \Phi\left(-\sqrt{\lambda}\epsilon\Omega\right)+\Phi\left(\sqrt{\lambda}(-2\Omega+\epsilon\Omega)\right) \right).
   \end{align}
\end{proof}
To characterize the effective cutoff frequency after centralization and Gaussian windowing, we introduce the following definition:
\begin{defn}\label{def: omega_win}
    Let $\sigma>0$ be the noise level, the effective cutoff frequency $\Omega_{win}$ is 
    \begin{align}\label{omega_win}
        \Omega_{win} := \Omega \cdot \left (1 - \inf_{0<\epsilon< 1} \left\{ \Phi\left(-\sqrt{\lambda}\epsilon\Omega\right)+\Phi\left(\sqrt{\lambda}(-2\Omega+\epsilon\Omega)\right)<\frac{\sqrt{\pi}}{\|\nu\|_{TV}}\sigma  \right\}\right).
    \end{align}
\end{defn}
\begin{rmk}
    $\Omega_{win}$ is well-defined since the function $H(\epsilon) = \Phi\left(-\sqrt{\lambda}\epsilon\Omega\right)+\Phi\left(\sqrt{\lambda}(-2\Omega+\epsilon\Omega)\right)$ is decreasing for $\epsilon\in(0,1)$. 
     Note that for any $\omega\in [-\Omega_{win},\Omega_{win}]$, $|\mathcal{E}_{mod}(\omega)|\le \sqrt{\frac{\pi}{\lambda}}$. To get a rough estimation of $\Omega_{win}$, one may combine Lemma \ref{lem: gauss-ineq} and the fact that $\Phi\left(-\sqrt{\lambda}\epsilon\Omega\right)>\Phi\left(\sqrt{\lambda}(-2\Omega+\epsilon\Omega)\right)$. By solving the equation $\frac{e^{-\lambda\epsilon^2\Omega^2}}{\sqrt{\lambda}\epsilon\Omega} = \frac{\sqrt{\pi}}{2\|\nu\|_{TV}}\sigma $, we get the value of $\epsilon$, an estimation of $\Omega_{win}$ can then be derived.
\end{rmk}
We point out that the Gaussian windowing causes the loss of high-frequency information, i.e. $\Omega_{win}<\Omega$. Moreover, the width of the Gaussian window and the effective cutoff frequency $\Omega_{win}$ are determined by choice of the Gaussian parameter, $\lambda$.
\begin{rmk}
One may consider other types of windowing functions, and all the analyses follow a similar framework to the one in this paper.
\end{rmk}

\subsection{SCAN-MUSIC}\label{subsec: scan-music}
Guided by the theoretical discussion in the previous subsection, we propose the SCAN-MUSIC in this subsection. In practice, we need to deal with the discrete version of (\ref{center-windowing}). We thus truncate the Gaussian window at some level $\gamma =\mathcal{O}(\sigma)>0$, and define
\begin{align}{\label{def:gamma}}
     \Gamma = \min\{ s\in\mathbb{N}: s\le K, e^{-\lambda(sh)^2}\le \gamma \}.
\end{align}
We also denote the centralized measurement data by
\begin{align}
    Y_{cen} = [\mathcal{S}_\mu Y (\omega_{-K}),\cdots,\mathcal{S}_\mu Y (\omega_{K})],
\end{align}
and the discrete truncated Gaussian window by
\begin{align}
    G_{\lambda,\Gamma} = h[G_\lambda(\omega_{-\Gamma}),\cdots,G_\lambda(\omega_{\Gamma})].
\end{align}
The numerical implementation of centralization and Gaussian windowing is given by
\begin{align}{\label{Y_win}}
    Y_{win} = Y_{cen} \ast G_{\lambda,\Gamma} [2\Gamma+1:2K+1],
\end{align}
where we only use the middle part of the discrete convolution to avoid the boundary effect of the convolution of two finite-length vectors.
More precisely, 
\begin{align}
Y_{win}[l] = \sum_{j=-\Gamma}^{\Gamma} Y_{cen}[\Gamma+l-j]\cdot G_{\lambda,\Gamma}[\Gamma+1+j], \quad l = 1,\cdots,2K-2\Gamma+1. 
\end{align}

The pseudo-code of centralization and Gaussian windowing is shown in \textbf{Algorithm \ref{algo_gauss}}. 
The equation (\ref{center-windowing}) implies that the processed measurement can be regarded as the signal generated by the centralized line spectra. According to the Nyquist-Shannon sampling criterion, we can enlarge the sampling step size to reconstruct the spectra therein. This fact motivates us to deploy a subsampling strategy, which can significantly reduce the computational cost. To implement the idea, we define a subsampling factor, $F_{sub}$, and enlarge the sampling step size to $F_{sub}\cdot h$. See Section \ref{subsec: detail1} for the choice of $F_{sub}$ and \textbf{Algorithm \ref{algo_sub1}} for the pseudo-code.\\

With the two techniques above, we are ready to introduce the algorithm for SCAN-MUSIC. 
First, for a given $\mu$, we apply centralization and Gaussian windowing for the chosen center $\mu$. Then the measurement can be regarded as the signal generated by the spectra within the essential region $[-R_{ess},R_{ess}] =: \{x\in \Real: e^{-\frac{x^2}{4\lambda}}\ge \kappa_E\}$, where $\kappa_E \sim \mathcal{O}(\sigma)$ is a parameter called the essential level. The essential region can be viewed as an approximation of the essential support of the Gaussian window.  
Second, we apply subsampling and MUSIC algorithm to reconstruct the spectra falling in the trust region $[-R_{tru}, R_{tru}]=:  \{x\in \Real: e^{-\frac{x^2}{4\lambda}} \ge \kappa_T\}$, where 
$\kappa_T$ is a parameter called the trust level which serves as an amplitude threshold for the centralized and windowed signal so that the spectral in $[-R_{tru}, R_{tru}]$ are less affected by the Gaussian windowing (recall the factor $e^{-\frac{(y_j-\mu)^2}{4\lambda}}$). In practice, we choose $\kappa_T$
to be close to $1$. 
Finally, the reconstruction is completed as $\mu$ sweep over the spectral domain.\\

The pseudo-code of the complete algorithm is displayed in \textbf{Algorithm \ref{algo_spectra}}. The implementation details can be found in Section \ref{subsec: detail1} and the complete error estimates for SCAN-MUSIC is shown in Section \ref{subsec: error analysis}.

\begin{algorithm}
	\caption{SCAN-MUSIC}
	\label{algo_spectra}
	\DontPrintSemicolon
    \SetKwInOut{Input}{Input}
    \SetKwInOut{Output}{Output}
    \SetKwInOut{Init}{Initialization}
	\Input{Original measurement: $Y$, noise level: $\sigma$, sample interval: $[-\Omega,\Omega]$, sampling step size: $h$, Sufficient large number: $H$} 
    \Input{Gaussian parameter: $\lambda$, Subsampling factor $F_{sub}$, spectral domain: $[R_1,R_2]$, Trust level: $L_{T}$, Essential level: $L_{E}$, Truncation level: $\gamma$}
        \tcc{determining trust region and essential support} 
		$grid \gets 0:h:H $ \;
		$\mathcal{F}G_{pre} \gets \exp(-\frac{1}{4\lambda}\cdot grid.^{\wedge} 2)$\;
		$\mathcal{F}G_{tru} \gets \mathcal{F}G_{pre}(\mathcal{F}G_{pre}>L_{T})$, $\mathcal{F}G_{ess} \gets \mathcal{F}G_{ess}(\mathcal{F}G_{ess}>L_{E})$\;
        $R_{tru} \gets [|\mathcal{F}G_{tru}|\cdot h]$,$R_{ess} \gets [|\mathcal{F}G_{ess}|\cdot h]$\;
        \tcc{reconstruction the line spectra in $[R_1,R_2]$} 
        $G_{center} \gets R_1 + R_{tru} : 2R_{tru} : R_2 - R_{tru}$\;
		\For{$m\ =\ 1:|G_{center}|$}{
            $\hat{y}_{pre} \gets [\ ]$\;
            $[Y_{win}] \gets CGM(Y,\Omega,h,\lambda,G_{center}(m),\gamma)$\;
            $[Y_{sub},h_{sub}]\gets Sub_1(Y_{win},\Omega_{\gamma},F_{sub},h)$\;
            $\hat{y}_{pre}$ is returned by MUSIC in the trust region $[-R_{tru},R_{tru}]$
            $\hat{y} \gets [\hat{y},\hat{y}_{pre}+G_{center}(m)]$\;
        }
	\Return $\hat{y}$\;
\end{algorithm}

\subsection{Sampling Complexity of SCAN-MUSIC}\label{subsec: sampling complexity}

From the mathematical model (\ref{measurement}), we notice that the sampling complexity is given by $2K+1$. The following theorem characterizes the relationship of sampling complexity required for SCAN-MUSIC and the source number in the regime $1/\rho\sim \mathcal{O}(1)$ and $\rho<1$. The general case is discussed in Section \ref{subsec: super-sparse regime}.
\begin{thm}\label{thm: random-complex}
    Assume that $supp(\nu) \subset [-R,R]$ and the sampling region is $[-\Omega,\Omega]$. In the regime $1/\rho\sim \mathcal{O}(1)$ and $\rho<1$, the sampling complexity $K$ required for SCAN-MUSIC and the source number satisfies
    \begin{align}
        K \ge \max\{n, \frac{n}{2\rho} \}.
    \end{align}
\end{thm}
\begin{proof}
    First, notice that there are $2n$ unknowns to determine $\nu$, it is clear that $K \ge n$.\\
    For $supp(\nu) \subset [-R,R]$, by the Nyquist-Shannon criterion, we have $ \frac{\Omega}{K}\le \frac{\pi}{R}$. Combining the above argument and the definition of $\rho$, we have 
    \begin{align}{\label{complexity-random}}
        K \ge   R \cdot\frac{\Omega}{\pi} 
        = \frac{n}{2\rho}.
    \end{align}
    Therefore,  $K \ge \max\{n, \frac{n}{2\rho} \}$.
\end{proof}
The above theorem shows that, for $1/\rho\sim \mathcal{O}(1)$, the sampling complexity of SCAN-MUSIC can be $K \sim \mathcal{O}(n)$. We point out that the above argument can also be understood in the framework of finite rate of innovation (FRI) \cite{1003065}. Notice that the rate of innovation of the line spectral is $\frac{n}{R}$, and the sampling complexity is then also linear to the rate of innovation.

\subsection{Computational Complexity}{\label{subsec: random-compute-complexity}}
Using SCAN-MUSIC, the problem of wide-band random spectra reconstruction is divided into $\lceil \frac{R}{R_{tru}} \rceil$ mutually independent subproblems. In each subproblem, the complexity consists of the following three parts: the implementation of centralization and Gaussian windowing, the MUSIC algorithm, and other operations including subsampling. \\

In the first part, the complexity of the centralization step is of the order $\mathcal{O}(K)$, and the complexity of convolution is of the order $\mathcal{O}\left( K\log K \right)$. It is clear that the complexity of the third part is of the order $\mathcal{O}(K)$. The most time-consuming part is the MUSIC algorithm. First, the SVD of a square matrix of order $[\frac{|Y_{win}|}{2F_{sub}}]$ is needed for the line spectra reconstruction by the MUSIC algorithm. The complexity therein is of the order $\mathcal{O}\left( \frac{|Y_{win}|^3}{F_{sub}^3} \right)$. Denote the number of the test points per unit interval by $N$. The complexity of calculating the MUSIC functional is of the order $\mathcal{O}\left( R_{tru}N \cdot \frac{|Y_{win}|^2}{F_{sub}^2} \right)$. From the above estimations, we conclude that the total complexity of SCAN-MUSIC is of the order
\begin{align}{\label{random-complex}}
    \mathcal{O}\left(\frac{R}{R_{tru}} \cdot\left(\frac{|Y_{win}|^3}{F_{sub}^3}+R_{tru}N \cdot \frac{|Y_{win}|^2}{F_{sub}^2}+K\log K \right)  \right).
\end{align}
On the other hand, using the prior information on the line spectral density, we can choose an appropriate $F_{sub}$ to make $\frac{|Y_{win}|}{F_{sub}}\sim R_{ess} \rho$ in practice (as shown in Section \ref{subsec: detail1}). Combining this observation and the definition of $\rho$, we can rewrite the result in (\ref{random-complex}) as follows
\begin{align}
    \mathcal{O}\left( n\cdot\left( (R_{ess} \rho)^3 + (R_{ess} \rho)^2R_{tru}N + n\log n \right)  \right).
\end{align}
Since $R_{ess} \rho$, $R_{tru}$, $N$ are fixed numbers, we conclude that the result in (\ref{random-complex}) can be rewritten as
\begin{align}{\label{re-random-complex}}
    \mathcal{O}\left( n^2\log n \right).
\end{align}

By a similar argument, we can calculate the complexity of the standard MUSIC algorithm of the order
\begin{align}{\label{music-random-complex}}
    \mathcal{O}\left(K^3+ RN \cdot K^2\right)\sim \mathcal{O}\left(n^3\right).
\end{align}
Comparing the results in (\ref{re-random-complex}) and (\ref{music-random-complex}), we see SCAN-MUSIC significantly reduces the computational cost of the standard MUSIC algorithm.

\subsection{Scaling Argument and Super-sparse Regime}\label{subsec: super-sparse regime}
In this section, we adopt a scaling argument to complete the discussion of the sampling complexity of SCAN-MUSIC. We denote the scaling parameter as $\tau$ with $0<\tau \le 1$. Observe that 
\begin{align}
    f(\omega) = \sum_{j=1}^n a_j e^{i y_j \omega} = \sum_{j=1}^n a_j e^{i \frac{y_j}{\tau}\cdot \tau \omega}.
\end{align}
For any $\rho<1$, we pick $0<\tau\le 1$ such that $\tau/\rho \sim \mathcal{O}(1)$. We adopt the downsampling strategy \cite{liu2022measurement}, which uses the data in the interval $[-\tau\Omega,\tau\Omega]$ to reconstruct the line spectra. We denote the minimum number of samples as $\tilde{K}$. It is clear that $\tilde{K}=2\tau K+1$. The dependence on the source number, $n$, is characterized in the following theorem. 
\begin{thm}\label{thm: random-complex-super}
    Assume $supp(\nu) \subset [-R,R]$ and sampling region $[-\tau\Omega,\tau\Omega]$. In the regime that $\rho<1$, we have 
    \begin{align}
        \tilde{K} \ge \max\{2n, \frac{\tau}{\rho}\cdot n +1\}.
    \end{align}
\end{thm}
\begin{proof}
    First, notice that there are $2n$ unknowns to determine $\nu$, it is clear that $\tilde{K} \ge 2n$.\\
    For $supp(\nu) \subset [-R,R]$, by the Nyquist-Shannon criterion, we have $ \frac{2\tau\Omega}{\tilde{K}-1}\le \frac{\pi}{R}$. Combining the above argument and the definition of $\rho$, we have 
    \begin{align}{\label{complexity-random-super}}
        \tilde{K} \ge  2\tau R \cdot\frac{\Omega}{\pi} +1 
        = \frac{\tau}{\rho}\cdot n +1.
    \end{align}
    Therefore, $\tilde{K} \ge \max\{2n, \frac{\tau}{\rho}\cdot n +1\}$.
\end{proof}
The above theorem shows that, for any $\rho<1$, by an appropriate guess of $\tau$ that $\tau/\rho \sim \mathcal{O}(1)$, $\tilde{K} \sim \mathcal{O}(n)$ samples are sufficient to reconstruct the line spectra. The above theorem can also be regarded as the generalization of Theorem \ref{thm: random-complex}. 
Notice that, for the downsampled data, the Rayleigh limit becomes $\frac{\pi}{\tau \Omega}$. It suggests that if the minimum separation distance, $d_{min}$, is known as prior information, we can choose the scaling parameter based on $\tau \sim \frac{\pi}{d_{min}\Omega}$. The sampling complexity can therefore be reduced.
\begin{rmk}
    The above result can also be understood in the framework of finite rate of innovation (FRI). Notice that the rate of innovation of the line spectral is $2\rho$, and the sampling complexity is then linear to the rate of innovation.
\end{rmk}

Meanwhile, we conclude that the computational complexity after introducing the scaling parameter is $\mathcal{O}(n^2\log n)$ by a similar argument as in the previous section.\\

We point out that the super-sparse regime i.e. $\rho \ll 1$ is a special case. Therefore, SCAN-MUSIC can also be applied to the Sparse Fourier Transform problem as a deterministic algorithm. Without loss of generality, we assume $\tau = 1$ unless otherwise mentioned in the rest of the paper.

\subsection{Trade-off of computational cost and resolution}\label{subsec: tradeoff}
Before discussing the trade-off of computational cost and loss of resolution of SCAN-MUSIC, we first review the computational resolution limit of the LSE problem.\\

In the recent work \cite{9410626}, the authors showed that the computational limit on support reconstruction of the LSE problem, which is defined as the minimum separation distance between the spectra that ensures a stable reconstruction of the spectral supports, depends on the cutoff frequency, the number of line spectral, and the signal-to-noise ratio (SNR) which is defined to be $SNR:= \frac{m_{min}}{\sigma}$ in the following way:
\begin{align}
    \mathcal{D}_{supp} \sim \mathcal{O}\left(\frac{\pi}{\Omega}\left(\frac{1}{SNR}\right)^{\frac{1}{2n-1}}\right).
\end{align}
The above result can be interpreted as follows. If $n$ is small, then super-resolution is achievable with sufficient SNR. When $n$ is large, however, hoping to achieve super-resolution through extremely high SNR is not practical. For moderate SNR, $\left(\frac{1}{SNR}\right)^{\frac{1}{2n-1}} \sim \mathcal{O}(1)$, when $n \gg 1$. In this case, therefore, the minimum separation distance for stable reconstruction should be of the order $\mathcal{O}\left( \frac{\pi}{\Omega}\right)$. We illustrate this phenomenon in the following example. We set $\Omega=1$ and uniformly align $10$ line spectra with a separation distance of $\pi$. Under the noise level of $\mathcal{O}(10^{-2})$, the resulting MUSIC plot, as depicted in Figure \ref{Fig: music plot}, exhibits only $8$ peaks. The situation does not improve much even if we reduce the noise level to $\mathcal{O}(10^{-5})$. This observation indicates that the original MUSIC algorithm fails to achieve successful reconstruction under this setup. Thus, the LSE problem in the regime where $\rho=1$ proves to be challenging when the spectra number is large. On the other hand, when the line spectra have cluster structures as described in section \ref{subsec: SCAN-MUSIC for cluster}, the ``global'' problem can be decoupled into ``local'' problems of reconstructing line spectra in each cluster. Then the theory of computational resolution limit can be applied to each ``local'' problem with $n$ being the source number in the involved cluster. Therefore it is still possible to reconstruct line spectra therein with separation distance below the Rayleigh limit when $n$ is not large. This is indeed the case by using the SCAN-MUSIC(C) method developed in the next section. See the numerical results in Section \ref{sec: numerical study}. \\

\begin{figure}[ht]
	\centering
    \subfloat[$SNR\sim\mathcal{O}(10^{-2})$]{
	\includegraphics[width=0.4\textwidth]{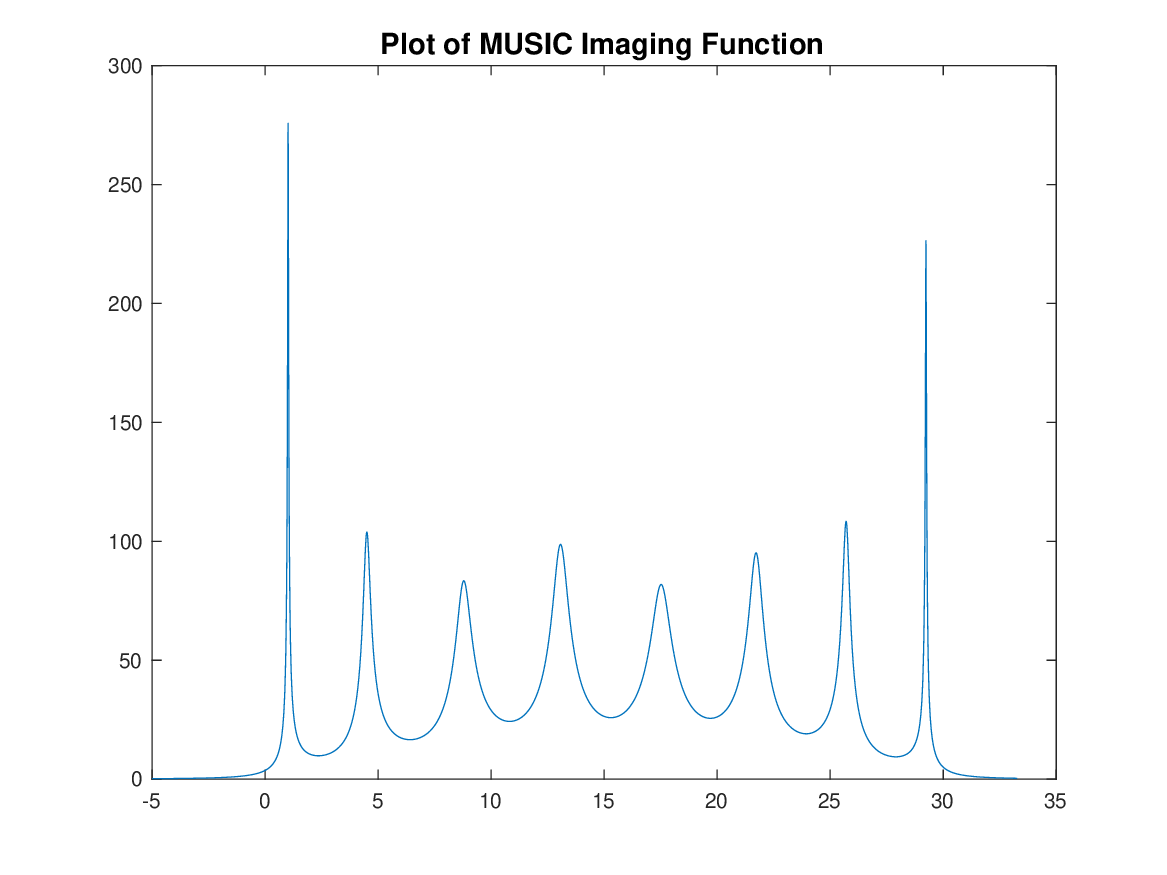}}
    \subfloat[$SNR\sim\mathcal{O}(10^{-5})$]{
	\includegraphics[width=0.4\textwidth]{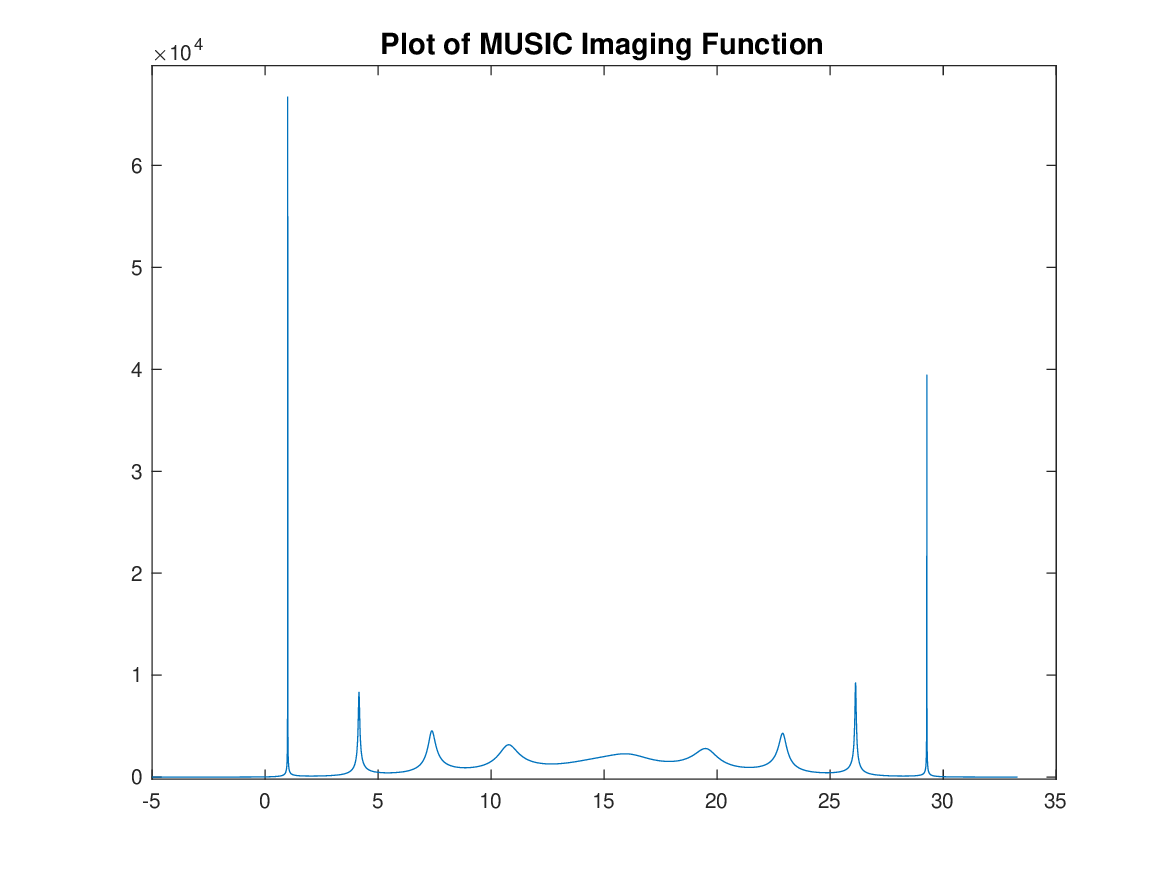}}
 \caption{MUSIC plot for 10 line spectra with uniform separation distance $\pi$ under different SNR.}
\label{Fig: music plot}
\end{figure}

For SCAN-MUSIC, we notice that under moderate noise, for any chosen center $\mu_0$, the number of line spectra within the essential region, $[-R_{ess},R_{ess}]$, is not small. Therefore, it is likely to have $\left(\frac{1}{SNR'}\right)^{\frac{1}{2n-1}} \sim \mathcal{O}(1)$, where we denote the $SNR'$ as the effective signal-to-noise ratio containing the original SNR and the computational error $\mathcal{E}_{total}$. Considering the resolution loss and downsampling discussed in Section \ref{subsec: scan-music}, the minimum separation distance for stable reconstruction is of the order $\mathcal{O}\left( \frac{\pi}{\Omega_{win}}\right)$.\\

We demonstrate the trade-off between computational complexity and loss of resolution by considering two settings. First, we consider the case $\lambda = +\infty$. Then SCAN-MUSIC coincides with the standard MUSIC method, which is time-consuming but with no loss of resolution in the reconstruction.  Second, we pick $\lambda$ to be sufficiently small such that the essential region has a length of only a few Rayleigh limits. Then, from the discussion in Section \ref{subsec: windowing and centralization}, $\Omega_{win}$ is significantly smaller than $\Omega$, which implies a loss of resolution. In this case, however, the number of line spectra within the essential support is small, and the subsampling strategy can therefore significantly reduce the computational cost. \\

We note that for a fixed choice of $\lambda$, a smaller choice of $\kappa_T$ includes more line spectral with less contrast to noise which lowers the stability of the algorithm. On the other hand, such a choice leads to fewer subproblems to be solved which increases the efficiency of the algorithm.

\section{Clustered spectra Reconstruction}\label{sec: cluster spectra}
In this section, we consider the reconstruction of the line spectra with cluster structure. In Section \ref{subsec: SCAN-MUSIC for cluster}, we introduce the model for line spectra with cluster structure. We then refine SCAN-MUSIC using the annihilating filter technique in Section \ref{subsec: cluster algorithm}. The sampling and computational complexity of the proposed method are shown in Section \ref{subsec: cluster complexity}. Finally, the reconstruction of cluster centers is discussed in Section \ref{subsec: cluster center}.\\

\subsection{SCAN-MUSIC Meets Clustered line spectra}\label{subsec: SCAN-MUSIC for cluster}
First, we introduce the mathematical model for line spectra with cluster structure. Let $\mathcal{I}_t \subset \Real$ be a closed interval centered at $\Bar{y}_t$ with half-length $D_t$, for $t = 1,\cdots, T$. We say $\cup_{t=1}^T \mathcal{I}_t$ is $(T,D,L)$-region if $\{\mathcal{I}_t\}$ are pairwise disjoint,
\begin{align}
    \max_{t=1,\cdots, T} D_j \le D,\quad \min_{1\le s<t\le T} |\Bar{y}_s-\Bar{y}_t|\ge L+2D.
\end{align}
We call the line spectra represented by $\nu$ is $(T,D,L)$-clustered if for some $(T,D,L)$-region, $\cup_{t=1}^T \mathcal{I}_t$, $supp(\nu)\subset \bigcup_{t=1}^T \mathcal{I}_t$ and $supp(\nu)\cap \mathcal{I}_t \ne \varnothing,\ \forall t = 1,\cdots, T$. We denote $\nu_t = \nu\cdot\chi_{\mathcal{I}_t}$ to be the $t$-th cluster and suppose there are $n_t$ line spectra therein with position $\{y_{t,s}\}$ and corresponding amplitude $\{a_{t,s}\}$. Assume that $n_t \le N_0$, for any $t = 1,\cdots,T$. We define the local measurement, $f_t$ by
\begin{align}
    f_t(\omega) = \sum_{s=1}^{n_t} a_{t,s}e^{i y_{t,s}\omega},\quad \omega\in[-\Omega,\Omega].
\end{align}
Then, the mathematical model in (\ref{measurement}) can be written into 
\begin{align}{\label{cluster measurement}}
    Y(\omega_k) = \sum_{t=1}^T f_t(\omega_k)+W(\omega_k) = \sum_{t=1}^T \sum_{s=1}^{n_t} a_{t,s}e^{iy_{t,s}\omega_k}+W(\omega_k),\ \omega_k\in[-\Omega,\Omega].
\end{align}
We start by considering the clustered structure with a single cluster of line spectra. In this case, we let $\mu$ be the cluster center and we conclude the computational limit of SCAN-MUSIC is of the order
\begin{align}
    \mathcal{O}\left(\frac{\pi}{\Omega_{win}}\left(\frac{1}{SNR}\right)^{\frac{1}{2n-1}}\right).
\end{align}

For the multiple cluster case, we expect that when the clusters are well-separated, the interference between different clusters will be small. SCAN-MUSIC should be able to achieve stable reconstruction. To confirm this idea, we design the following numerical experiment.\\

We set $\Omega = 1$, $\sigma = 10^{-3}$, $h = 0.001$. Then, the corresponding Rayleigh limit is $\pi$. We set up 50 clusters with the distance between each cluster being $6 \pi$, and there are two line spectra in each cluster with separation distance $1$. According to the distribution of the line spectra over the spectral domain, we pick the Gaussian parameter, the trust level, and the subsampling factor as $(\lambda, \kappa_T, F_{sub}) = (100,0.95,50)$. The absolute error and running time are plotted in Figure \ref{Fig: random-cluster}.\\
\begin{figure}[ht]
	\centering
    \subfloat[Running time]{
	\includegraphics[width=0.4\textwidth]{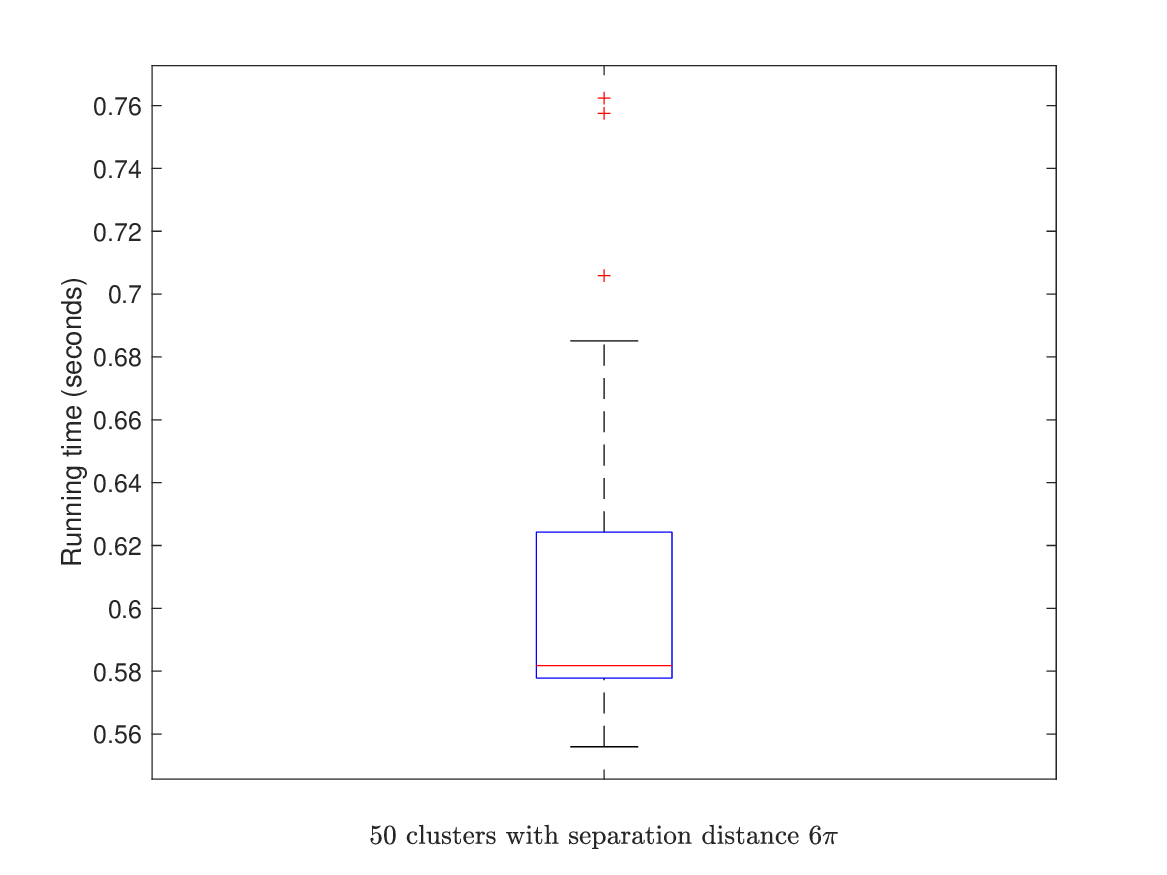}}
    \subfloat[Reconstruction error]{
	\includegraphics[width=0.4\textwidth]{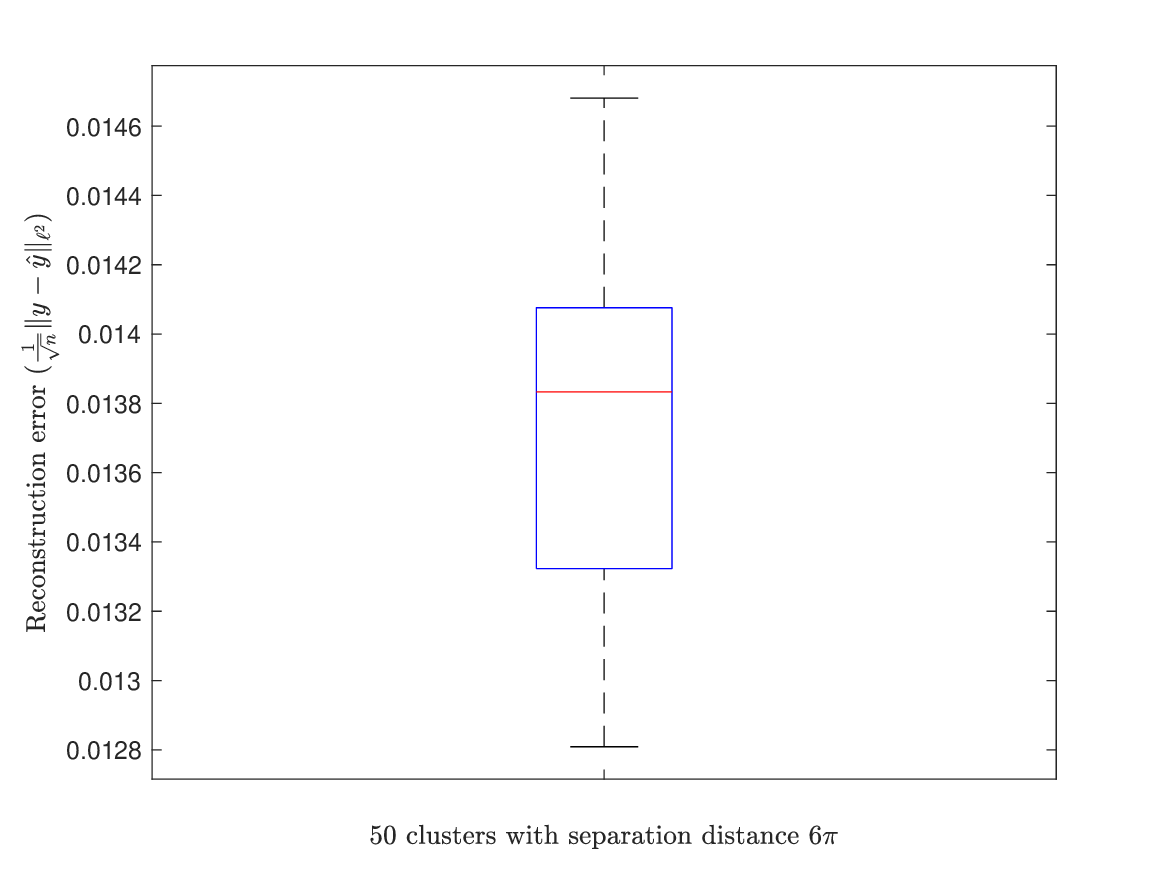}}
 \caption{Numerical experiment result of SCAN-MUSIC method for clustered line spectra.}
\label{Fig: random-cluster}
\end{figure}

We observe that in the setup above, SCAN-MUSIC shows the super-resolution ability. We point out that when the separation distance for clusters is small, the interference between different clusters becomes large, which leads to the failure of stable reconstruction. In the next section, the tools of multipole expansion and annihilating filters are introduced to deal with this issue.

\subsection{Annihilating Filter Based Cluster Removal and Localization}\label{subsec: cluster algorithm} 

In this section, we first introduce the multipole expansion and annihilating filter. Refining SCAN-MUSIC by the annihilating filter technique, we propose the SCAN-MUSIC for clustered spectra, called SCAN-MUSIC(C).\\

The annihilating filter is proposed in the framework of finite rate of innovation theory \cite{1003065}. Literature usually achieves spectra reconstruction with the help of annihilating filters, see e.g. \cite{7547372}\cite{doi:10.1137/15M1042280}. In recent work \cite{fei2023iff}, the authors achieve source removal using the annihilating filter technique. 
We start with the following multipole expansion \cite{9410626}\cite{LIU2022402} of $f_t$,
\begin{align}
    f_t(\omega) 
    &= \sum_{s=1}^{n_t} a_{t,s}e^{i y_{t,s}\omega} 
    = \sum_{s=1}^{n_t} a_{t,s}e^{i\Bar{y}_t\omega}e^{i (y_{t,s}-\Bar{y}_t)\omega} \notag\\
    &=\sum_{s=1}^{n_t} a_{t,s}e^{i\Bar{y}_t\omega}\sum_{r=0}^{\infty} \frac{(y_{t,s}-\Bar{y}_t)^r}{r!} (i\omega)^r\notag\\
    &=\sum_{r=0}^{\infty}\sum_{s=1}^{n_t} a_{t,s} \frac{(y_{t,s}-\Bar{y}_t)^r}{r!} \cdot (i\omega)^re^{i\Bar{y}_t\omega}
\end{align}
Note that 
\begin{align}
    \mathcal{F}\left((i\omega)^re^{i\Bar{y}_t\omega}\right) = \delta_{\Bar{y}_t} ^{(r)},
\end{align}
and recall that $h = \frac{\Omega}{K}$ is the sampling step size. Let $A_{t,M_t}$ be the $M_t$-th order annihilating polynomial defined by 
\begin{align}\label{filter spatial domain}
    A_{t,M_t}(x) = \left(e^{ihx}-e^{ih\Bar{y}_t}\right)^{M_t} = \sum_{l = 0}^{M_t} c_{t,l} e^{ilhx},
\end{align}
where $M_t$ is a positive integer with $M_t>r$. Then, straightforward calculation shows that
\begin{align}
    \sum_{l = 0}^{M_t} c_{t,l} \cdot (i\omega_{k-l})^re^{i\Bar{y}_t\omega_{k-l}} = 0, \quad \forall k = -K+M_t,\cdots,K.
\end{align}
The above identity can be interpreted in the spatial domain as
\begin{align}
    \delta_{\Bar{y}_t} ^{(r)} \cdot  A_{t,M_t}(x) = 0,
\end{align}
in the sense of distribution. The $M_t$-th order annihilating filter applied at cluster center $\Bar{y}_s$ can be derived based on (\ref{filter spatial domain}) as
\begin{align}\label{filter frequency domain}
    Q_t = \underbrace{[1,-e^{i\Bar{y}_t h}]\ast \cdots \ast [1,-e^{i\Bar{y}_t h}]}_{M_t\  times} = [c_{t,0},c_{t,1},\cdots,c_{t,M_t}].
\end{align}
Recall that $D_t$ denotes the half-length of the interval $\mathcal{I}_t$. The following proposition in the spacial domain characterizes decaying property for the profile of $t$-th cluster after applying the annihilating filter.
\begin{prop}{\label{choice of poles}}
For the local measurement $f_t$, we have
\begin{align}{\label{filtered image estimation}}
    \left| \left(\mathcal{F} f_t \cdot A_{t,M_t}\right) \ast \left(\mathcal{F}\chi_{[-\Omega,\Omega]}\right)(x)\right|
    \le \frac{\|\nu_t\|_{TV}}{\pi}\cdot (hD_t)^{M_t} e^{hD_t}\cdot\frac{1}{|x-\Bar{y}_t|}.
\end{align}
\end{prop}
\begin{proof}
    First, we notice that for any given $M_t,r\ge 0$, we have
\begin{align}
    \delta_{\Bar{y}_t}^{(r)}\cdot A_{t,M_t}(x) = \delta_{\Bar{y}_t}^{(r)}\cdot \left( e^{ihx}-e^{ih\Bar{y}_t}\right)^{M_t}=
    \begin{cases}
        0, & r<M_t,\\
        (-ih)^r M_t! e^{ih\Bar{y}_t}\delta_{\Bar{y}_t}, & r\ge M_t.
    \end{cases}
\end{align}
Then, we have
\begin{align}
    \mathcal{F}f_t \cdot A_{t,M_t} 
    &= \left( \sum_{r=0}^{\infty} \sum_{s=1}^{n_t} a_{t,s}\frac{(y_{t,s}-\Bar{y}_t)^r}{r!}\cdot \delta_{\Bar{y}_t}^{(r)}\right)\cdot A_{t,M_t} \notag\\
    &= \sum_{r=M_t}^{\infty} \sum_{s=1}^{n_t} a_{t,s}\frac{(y_{t,s}-\Bar{y}_t)^r}{(r-M_t)!}\cdot (-ih)^re^{ih\Bar{y}_t}\delta_{\Bar{y}_t}.
\end{align}
Therefore,
\begin{align}
    |\left(\mathcal{F} f_t \cdot A_{t,M_t}\right) \ast \mathcal{F}\chi_{[-\Omega,\Omega]}(x)|
    &= |\sum_{r=M_t}^{\infty} \sum_{s=1}^{n_t} a_{t,s}\frac{(y_{t,s}-\Bar{y}_t)^r}{(r-M_t)!}\cdot (-ih)^re^{ih\Bar{y}_t}p(x-\Bar{y}_t)|   \notag\\
    &\le \sum_{r=M_t}^{\infty} \sum_{s=1}^{n_t} |a_{t,s}|\cdot \frac{(hD_t)^r}{(r-m)!}\cdot|p(x-\Bar{y}_t)| \notag\\
    &\le \frac{\|\nu_t\|_{TV}}{\pi}\cdot (hD_t)^{M_t} e^{hD_t}\cdot\frac{1}{|x-\Bar{y}_t|}
\end{align}
\end{proof}

The above proposition can be interpreted as follows. Since the sampling step size $h\leq \frac{\pi}{R}$ according to the Nyquist-Shannon criterion, $hD_t\ll 1$. Hence, the annihilating filter indeed filters a major part of the signal generated by the $t$-th cluster and can therefore reduce the interference to the other clusters. We observe that according to the result in (\ref{filtered image estimation}), larger $M_t$ can result in better annihilating effects. On the other hand, the application of an annihilating filter has an effect on the signal-to-noise ratio, which can be seen in the following calculation. Suppose that an $M_t$-th order annihilating filter is applied at $\Bar{y}_t$. Recall the annihilating filter defined in (\ref{filter frequency domain}). We have
\begin{align}
    |\sum_{l=0}^{M_t} c_{m,l}W(\omega_{k-l})| 
    \le \sigma \cdot \sum_{l=0}^{M_t} |c_{m,l}| = 2^{M_t}\sigma,  \quad \forall k = -K+M_t,\cdots,K, 
\end{align}
which implies that the annihilating filter step may reduce the SNR of the signals. Therefore, the order of the annihilating filter should be chosen to balance the annihilating ability and the loss of SNR.\\

The subsampling strategy for SCAN-MUSIC(C) is similar to the one for SCAN-MUSIC. The length of the subsampled signal should be long enough to enable the application of the MUSIC algorithm after the application of the annihilating filter. See the discussion in Section \ref{subsec: detail2}.\\

We now introduce the algorithm of SCAN-MUSIC(C). Assuming we have the prior information about the cluster center, $\{\Bar{y}_t\}_{t=1}^T$, the acquirement of this information is discussed in Section \ref{subsec: cluster center}. For cluster $\nu_t$, we set $\mu = \Bar{y}_t$. Similar to SCAN-MUSIC, we pick $\kappa_T$, and $\kappa_E$ as the trust level and essential level respectively. Here, $\kappa_T$ should be chosen such that $[-D_t,D_t]\subset [-R_{tru},-R_{tru}]$ and $\kappa_E \sim \mathcal{O}(\sigma)$. We apply the centralization and Gaussian windowing as in (\ref{center-windowing}). After that, we start constructing the annihilating filter to remove the interference caused by other clusters. We denote the set of cluster centers where the annihilating filter shall be applied as $S_t$, defined by
\begin{align}\label{S_t}
    S_t := \{\Bar{y}_s \in \{\Bar{y}_t\}_{t=1}^T : \Bar{y}_s-\mu \in [-R_{ess},R_{ess}]\setminus [-R_{tru},-R_{tru}] \} \triangleq \{ \Bar{y}_{t_1},\cdots\Bar{y}_{t_{|S_t|}}\}.
\end{align}

The whole annihilating filter is then defined as 
\begin{align}
    Q = Q_{t_1} \ast \cdots \ast Q_{t_{|S_t|}},
\end{align}
and the implementation of filtering is given by
\begin{align}
Y_{filt} = \left(Y_{win} \ast Q \right)\left[|Q|+1:|Y_{win}|\right].
\end{align}
Here, we choose the middle part of the discrete convolution of $Y_{win}$ and $Q$ to avoid the boundary effect. \\

We expect the filtered data $Y_{filt}$ to contain signals generated only by line spectra in a single cluster. The problem then becomes narrow-band. This motivates us to apply the subsampling strategy. In comparison to the one in SCAN-MUSIC, we can choose a larger subsampling factor, which reduces not only the computational burden of the localization step but also the signal-to-noise ratio degradation caused by the annihilation filter, see the discussion in Section \ref{subsec: detail2}. The MUSIC algorithm is then used for the localization step. Finally, by scanning over all the cluster centers, the complete reconstruction can be achieved. The discussion on the reconstruction of cluster centers can be found in Section \ref{subsec: cluster center}.\\

We display the pseudo-code of SCAN-MUSIC(C) in \textbf{Algorithm \ref{algo_cluster_spectra}}. From the structure of SCAN-MUSIC(C), we emphasize that the centralization and Gaussian windowing steps can be easily paralleled. The pseudo-code of subsampling and filtering are shown in \textbf{Algorithm \ref{algo_subsample_cluster}} and \textbf{Algorithm \ref{algo_filter}} respectively.

\begin{algorithm}
	\caption{SCAN-MUSIC(C)}
	\label{algo_cluster_spectra}
	\DontPrintSemicolon
    \SetKwInOut{Input}{Input}
    \SetKwInOut{Output}{Output}
    \SetKwInOut{Init}{Initialization}
	\Input{Original measurement: $Y$, noise level: $\sigma$, sample interval: $[-\Omega,\Omega]$, sampling step size: $h$, Sufficient large number: $H$.} 
    \Input{Gaussian parameter: $\lambda$, spectral domain: $[R_1,R_2]$, Trust level: $L_{T}$, Essential level: $L_{E}$, Cluster center: $\Bar{y}$, line spectra number in each cluster: $N_0$,  Truncation level: $\gamma$.}
        \tcc{determining trust region and essential support} 
		$grid \gets 0:h:H $ \;
		$\mathcal{F}G_{pre} \gets \exp(-\frac{1}{4\lambda}\cdot grid.^{\wedge} 2)$\;
		$\mathcal{F}G_{tru} \gets \mathcal{F}G_{pre}(\mathcal{F}G_{pre}>L_{T})$, $\mathcal{F}G_{ess} \gets \mathcal{F}G_{ess}(\mathcal{F}G_{ess}>L_{E})$\;
        $R_{tru} \gets [|\mathcal{F}G_{tru}|\cdot h]$,$R_{ess} \gets [|\mathcal{F}G_{ess}|\cdot h]$\;
        \tcc{reconstruction the clustered line spectra in $[R_1,R_2]$} 
		\For{$m\ =\ 1:|\Bar{y}|$}{
            $Y_{win} \gets CGM(Y,\Omega,h,\lambda,\Bar{y}(m),\gamma) $\;
            $y_{center} \gets \Bar{y}(R_{tru}<\Bar{y}-\Bar{y}(m)<R_{ess})$, $S_m$ is returned according to (\ref{S_t})\;
            $Ord_m\gets(M_{t_1},\cdots,M_{t_{|S_m|}}) $, where $\{M_j\}_{j=1}^{|S_m|}$ are returned by some prior guess. \;
            $[Y_{sub},h_{sub}] \gets Sub_2(Y_{win},y_{center},Ord_m,h,N_0)$\;
            $Y_{filt} \gets AFSR(Y_{sub}, y_{center}, Ord_m, h_{sub} )$\;
            $\hat{y}_{pre}$ is returned by MUSIC in the trust region $[-R_{tru},R_{tru}]$ using $Y_{filt}$\;
            $\hat{y} \gets [\hat{y},\hat{y}_{pre}+\Bar{y}(m)]$\;
        }
	\Return $\hat{y}$\;
\end{algorithm}

\subsection{Sampling and Computational Complexity}\label{subsec: cluster complexity}

Notice that the structure of SCAN-MUSIC(C) is similar to SCAN-MUSIC. The two algorithms share the same sampling requirement,\begin{align}\label{cluster sample complex}
    K \sim \mathcal{O}(n).
\end{align}
For computational complexity, we assume that the line spectra are $(T, D, L)$-clustered. Recall that the number of spectra within a single cluster is bounded by $N_0$ and that the reconstruction problem is divided into $T$ subproblems. For each subproblem, the computational consists of three parts. First, the implementation of centralization and Gaussian windowing. Second, the implementation of subsampling and annihilating filters. Finally, the application of the MUSIC algorithm for reconstruction.\\

For the first part, the complexity is of the order $\mathcal{O}(K\log K)$. For the second part, we notice that to filter the line spectra in a single cluster, $\mathcal{O}(N_0)$ multipoles should be used. Furthermore, the subsampling step has at most $\mathcal{O}(K)$ complexity. Since $|Y_{sub}|\sim \mathcal{O}(N_0)\ll K$, the complexity of the filtering step can be neglected. After filtering, the MUSIC algorithm is applied to a square matrix of order $|Y_{filt}|\approx 2N_0$, which has a computational complexity  
\begin{align}
    \mathcal{O}\left( N_0^3 + R_{tru}NN_0^2\right),
\end{align}
where $N$ is the number of test points per unit interval. Overall, combined with (\ref{cluster sample complex}), the total computational complexity of SCAN-MUSIC(C) is of the order
\begin{align}\label{cluster-complex}
    \mathcal{O}\left(T\cdot( K\log K +N_0^3 + R_{tru}NN_0^2)\right).
\end{align}
Comparing the result above with the one in (\ref{random-complex}), we observe that the computational cost of SCAN-MUSIC(C) is significantly less than the one of SCAN-MUSIC. Since $N_0$, $R_{tru}$, $N$ are all fixed numbers and the sampling complexity is linear in $n$, we rewrite (\ref{cluster-complex}) as 
\begin{align}\label{re-cluster-complex}
    \mathcal{O}\left(T n\log n \right).
\end{align}

\subsection{Cluster Center Reconstruction}\label{subsec: cluster center}
In this section, we address the estimation of cluster centers that are needed in the implementation of SCAN-MUSIC(C).\\

In the regime where the separation distance between clusters is multiple of the Rayleigh limit, i.e. $L \gtrsim \mathcal{O}\left(\frac{\pi}{\Omega}\right)$, 
we can use SCAN-MUSIC to reconstruct the cluster centers. The strategy is as follows. First, we apply SCAN-MUSIC to estimate the cluster center. In this step, since we do not require a high-resolution reconstruction, we can choose relatively smaller $\lambda$ to accelerate the algorithm. Then, we cluster the closely separated estimated line spectra, especially the ones that are separated within one Rayleigh length. After that, we take the average over each cluster and get the estimation of cluster centers. The choice of parameters for applying SCAN-MUSIC for cluster center reconstruction is discussed in Section \ref{subsec: detail1}. The following experiments are designed to test the efficiency and accuracy of SCAN-MUSIC for the cluster center reconstruction.\\

We set the $\Omega = 1$, $L = 4\pi$, noise level is of the order $\mathcal{O}(10^{-3})$. Let there be $T$ clusters separated in the spectral domain and in each cluster, there are two line spectra with a separation distance of $1$. The sampling step size for each setup is according to the Nyquist-Shannon sampling step size. For SCAN-MUSIC, we pick the Gaussian parameter, the trust level, and the subsampling factor as $(\lambda,\kappa_T, F_{sub})=(70,0.9,15)$. We conducted 20 random experiments for each setup.
The running time and reconstruction error of SCAN-MUSIC for cluster center reconstruction is shown in Figure \ref{Fig: cluster_detect}.
We observe that SCAN-MUSIC achieves stable reconstruction in a very efficient way. \\

\begin{figure}[ht]
	\centering
    \subfloat[Running time of reconstruction for different cluster number.]{
	\includegraphics[width=0.4\textwidth]{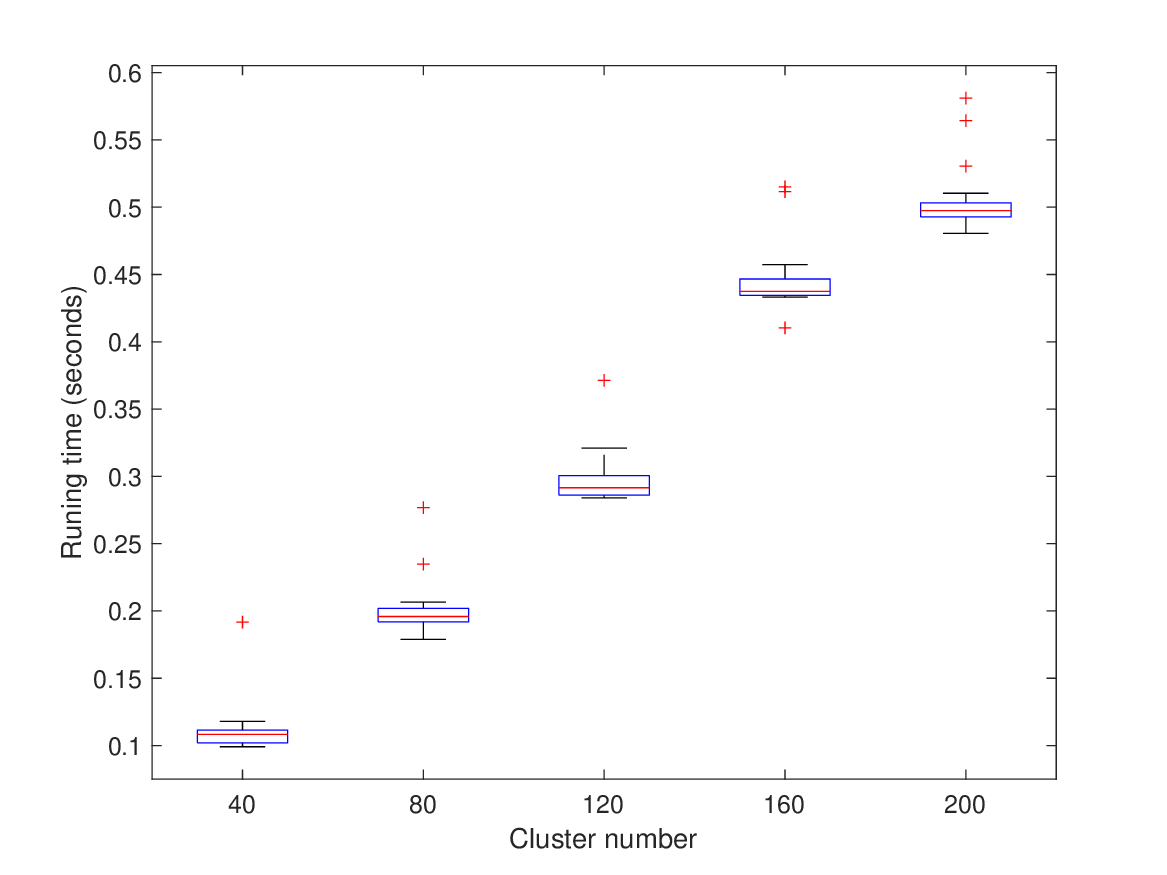}}
    \quad
    \subfloat[Reconstruction error for different cluster number.]{
	\includegraphics[width=0.4\textwidth]{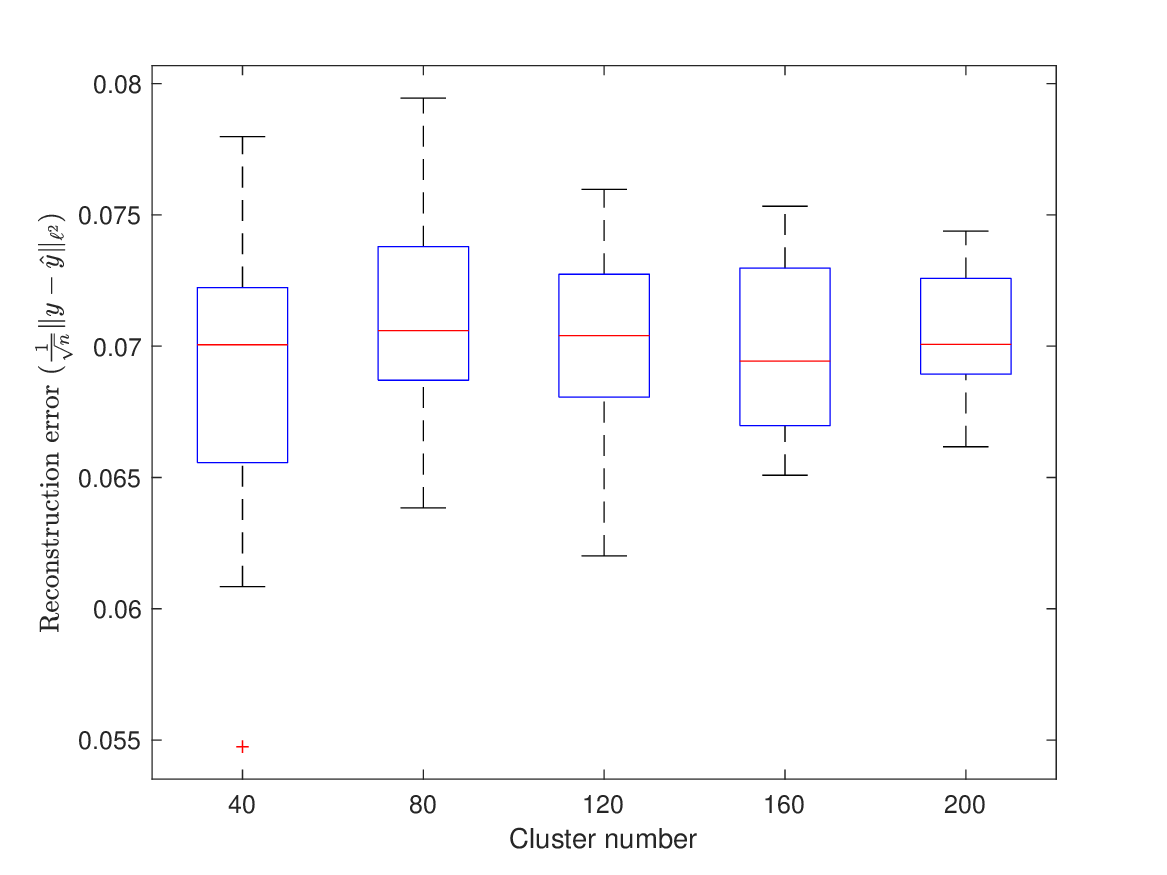}}
 \caption{Numerical result of SCAN-MUSIC for cluster center reconstruction}
\label{Fig: cluster_detect}
\end{figure}

We point out that other methods can be applied to detect the cluster structure in the regime that $L \gg \frac{\pi}{\Omega}$ i.e. the clusters are far-away separated.  For instance, the downsampling strategy proposed in \cite{liu2022measurement}, or FFT. If we further assume that the line spectra have positive amplitudes, the strategy in \cite{li2023note} can also be applied.

\section{Implementation Details and Extensions}\label{sec: details}
In this section, we provide more details
on the implementation of the proposed algorithms, especially on the choice of parameters. 

\subsection{Details of SCAN-MUSIC}\label{subsec: detail1}

As discussed in the previous sections, we choose the truncation level $\gamma\sim\mathcal{O}(\sigma)$, the essential level $\kappa_E\sim\mathcal{O}(\sigma)$, and the trust level $\kappa_T$ close to 1. \\

The choice of $\lambda$ determines the resolution of SCAN-MUSIC and $\lambda$ is tuned based on the trade-off between computational cost and the resolution, see Section \ref{subsec: tradeoff}. In practice, we can choose $\lambda\sim 10^2$. When SCAN-MUSIC is used for cluster detection, $\lambda$ can be smaller to reduce the computation cost (can also be the same for simplicity).\\ 

The essential level $R_{ess}$ can be derived from $\lambda$ and $\kappa_E$.  The subsampling factor $F_{sub}$ can be chosen as follows. Observe that the number of spectra in the essential region approximately equals $2R_{ess}\rho$ and therefore $4R_{ess}\rho$ samples are needed to apply MUSIC to reconstruct the spectra therein. Combining this observation and the Nyquist-Shannon criterion, we may choose $F_{sub}= \min \{\frac{|Y_{win}|}{4R_{ess}\rho}, \frac{\pi}{R_{ess}h}\}$.\\

\subsection{Details of SCAN-MUSIC(C)} \label{subsec: detail2}

The choice of $\lambda$, $\gamma$, $\kappa_E$ and $\kappa_T$ is similar to SCAN-MUSIC. For the choice of multipole order to apply annihilating filters, we refer to \cite{liu2022measurement} for the theoretical results for clustered line spectra in the frequency domain and we leave the complete theoretical discussion in the future work. Our numerical experience suggests that we can choose the multipole order to be 2 to 4 for the cluster consisting of 2 spectra, and 3 to 6 for the cluster consisting of 3 spectra. Specifically, to reconstruct the $t$-th cluster, it is necessary to assign a higher multipole order to the adjacent clusters. The intuition behind this is that the adjacent clusters have stronger interference to $t$-th cluster than those far away.\\

The subsampling strategy for SCAN-MUSIC(C) is similar to the one for SCAN-MUSIC. We need to ensure that the length of the subsampled signal is sufficient to apply the MUSIC algorithm after the filtering step. The subsampling factor can be chosen to be $F_{sub} = \lceil \frac{|Y_{win}|}{2N_0+|Q|} \rceil$, where $N_0$ is the maximum number of spectra number within a cluster, and $|Q|$ is the length of annihilating filter.



\section{Numerical Study}\label{sec: numerical study}
In this section, we conduct several groups of experiments to test the numerical behavior of the two proposed methods. We set $\Omega=1$ which implies the corresponding Rayleigh limit is $\pi$. For each different setup, we conduct $10$ independent random experiments. In this section, we do not perform parallel computing for the algorithms proposed in this paper.
\subsection{Random Spectra Reconstruction}
We present two groups of numerical experiments in this section to demonstrate the numerical performance of SCAN-MUSIC.\\

We first test and compare the efficiency of traditional MUSIC and SCAN-MUSIC for different signal ranges. We set the noise level to be $\mathcal{O}(10^{-2})$. To avoid the unstable reconstruction of the MUSIC algorithm, the line spectra are randomly generated in the interval $[0, R]$ with separation distances $5$ to $10$. For sampling step size, we pick $h = \frac{3}{R}$, which is near the Nyquist-Shannon sampling step size. For SCAN-MUSIC in different setups, we fix the Gaussian parameter, trust level, and truncation level as $(\lambda, \kappa_T, \gamma) = (170, 0.95, 10^{-2})$. The subsampling factor is then chosen according to the density of line spectra. Figure \ref{Fig: random-music} shows the averaged running time of the traditional MUSIC algorithm and SCAN-MUSIC. We observe that SCAN-MUSIC significantly reduces the computational cost.\\

\begin{figure}[ht]
	\centering
    \subfloat[Running time of SCAN-MUSIC for different signal range.]{
	\includegraphics[width=0.4\textwidth]{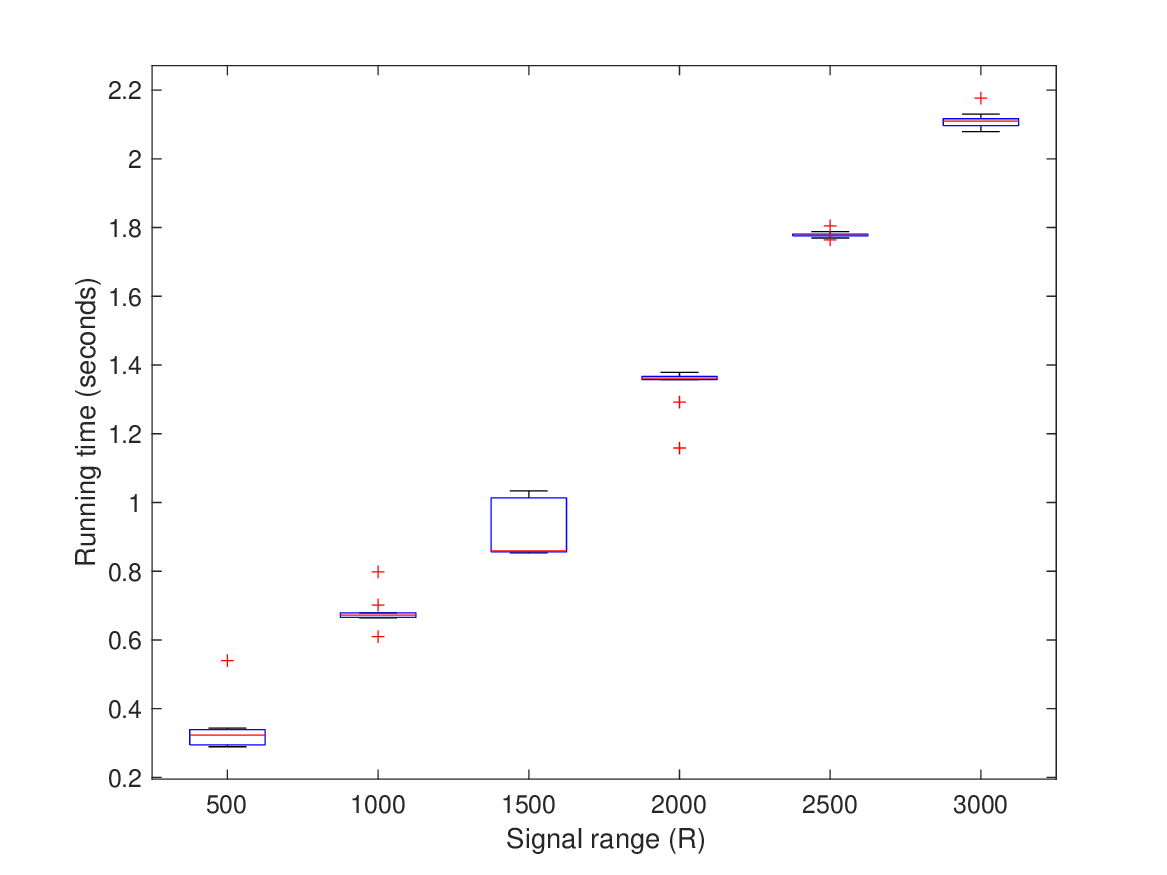}}
    \quad
    \subfloat[Averaged running time of MUSIC and SCAN-MUSIC for increasing signal range.]{
	\includegraphics[width=0.4\textwidth]{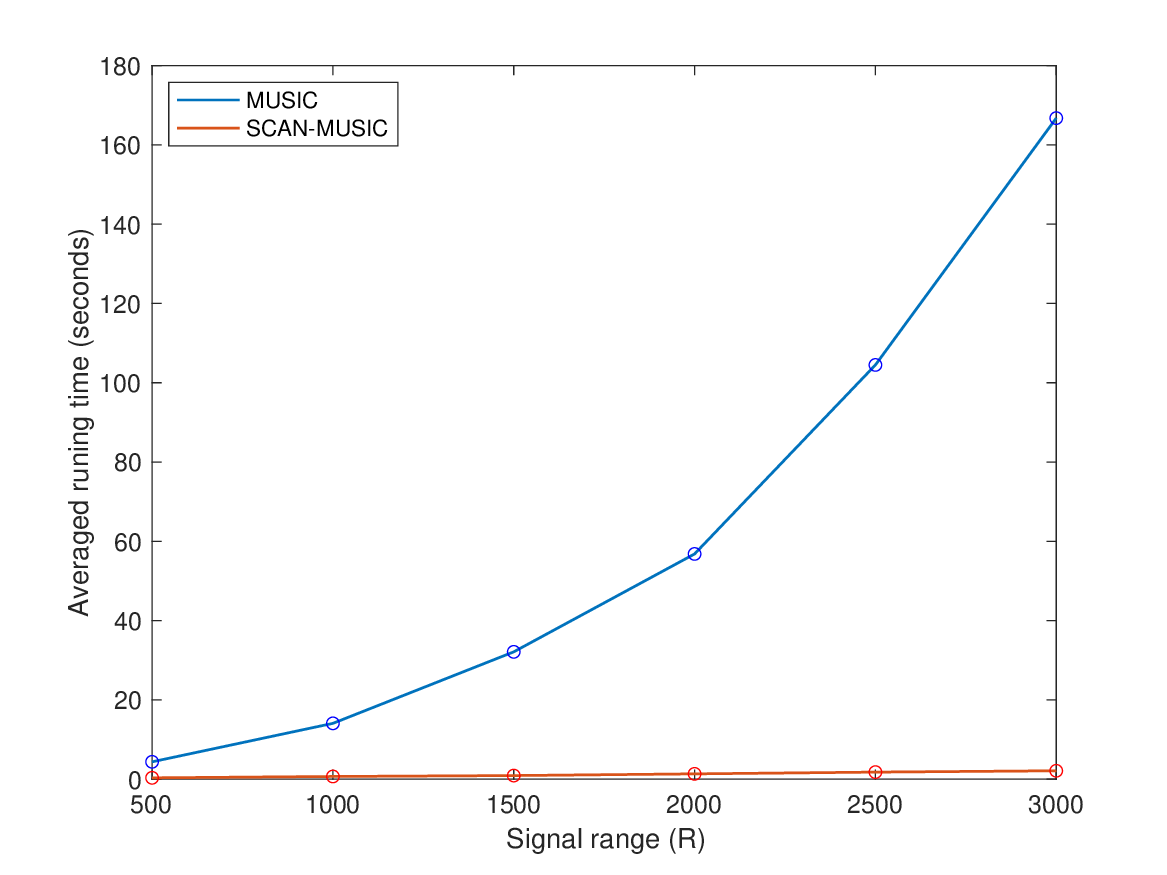}}
 \caption{Plot of the averaged running time of SCAN-MUSIC method with different signal range and the comparison with MUSIC algorithm.}
\label{Fig: random-music}
\end{figure}
We then test the stability of SCAN-MUSIC under different noise levels. We set the signal range $R=1000$. The line spectra are generated in the same way as in the previous experiment. We pick $h = 0.003$ and fix the Gaussian parameter $\lambda=170$, the trust level $\kappa_T = 0.95$. We choose the truncation level $\gamma = \sigma$, with $\sigma$ being the varying noise level in each group of experiments. The subsampling factor is chosen according to the spectra density as well as $\gamma$. We define the reconstruction error as $\frac{1}{\sqrt{n}}\|y-\hat{y}\|_{\ell^2}$. We present the numerical result in Figure \ref{Fig: different SNR}. We observe that the SCAN-MUSIC returns a stable estimation of the line spectra under the tested noise levels.
\begin{figure}
\centering
\includegraphics[width=0.7\textwidth]{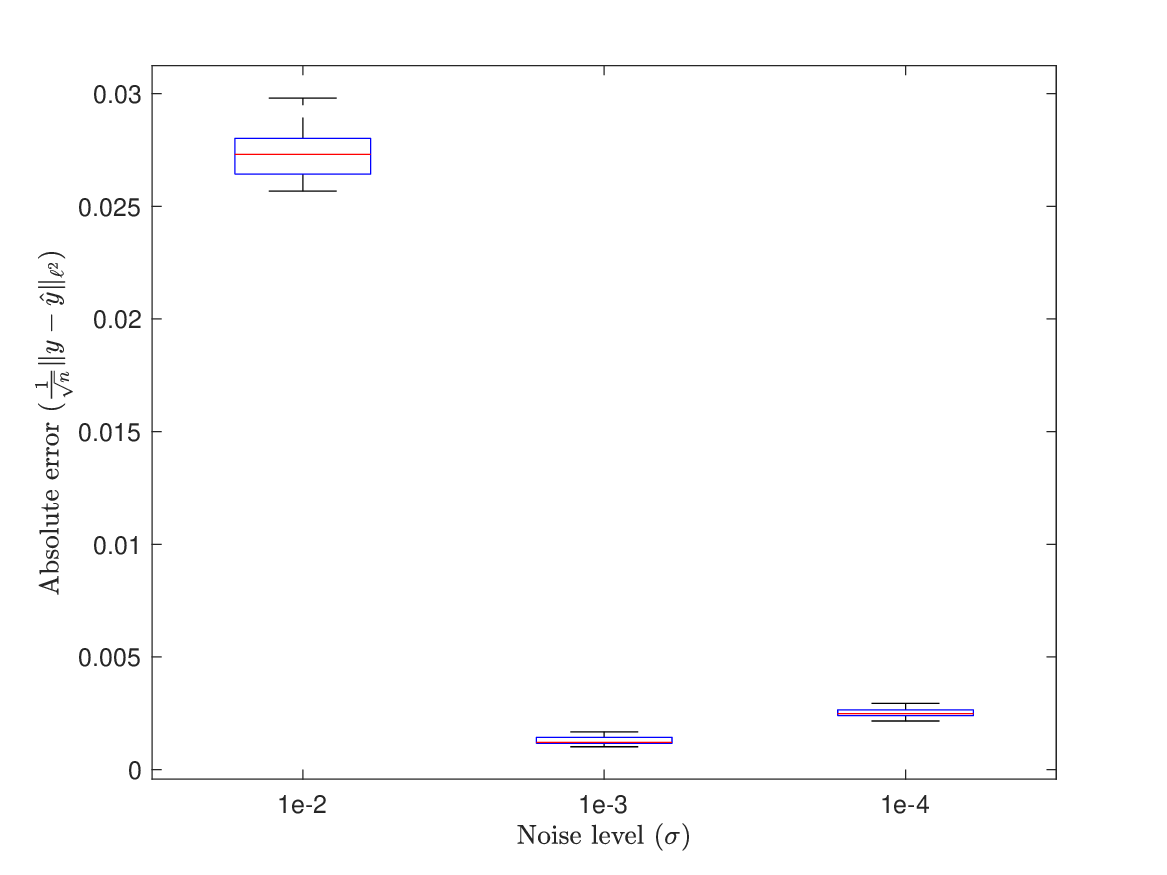}
\caption{Boxplot of reconstruction error of \textbf{Algorithm \ref{algo_spectra}} under different noise level.}
\label{Fig: different SNR}
\end{figure}

\subsection{Clustered Spectra Reconstruction}
In this section, we present two groups of numerical experiments to demonstrate the behavior of SCAN-MUSIC(C). We first apply SCAN-MUSIC to detect the cluster centers and apply SCAN-MUSIC(C) to reconstruct the line spectra. Since the running time of SCAN-MUSIC(C) is shorter than the running time of SCAN-MUSIC, we do not compare the running time of SCAN-MUSIC(C) and traditional MUSIC in this section.\\

For the first group, we set 2 line spectra in each cluster, the separation distance of the two line spectra is 1, and the distance between clusters is $4\pi$. We set $\Omega=1$, $\sigma = 10^{-3}$ and $h=0.001$. In both cluster center reconstruction and spectra reconstruction, we pick the Gaussian parameter, the trust level, and the subsampling factor as $(\lambda,\kappa_T, F_{sub})=(70,0.9,60)$ for the cluster center reconstruction. For the spectra reconstruction, we pick the multipole order for the clusters by the following strategy: we choose the multipole order to be 3 for the nearest cluster and 2 for others. The running time of SCAN-MUSIC(C) combined with SCAN-MUSIC and the reconstruction error for the reconstructed position is shown in Figure \ref{Fig: random-time-2}.\\
\begin{figure}[ht]
	\centering
    \subfloat[Running time of reconstruction for different cluster number.]{
	\includegraphics[width=0.4\textwidth]{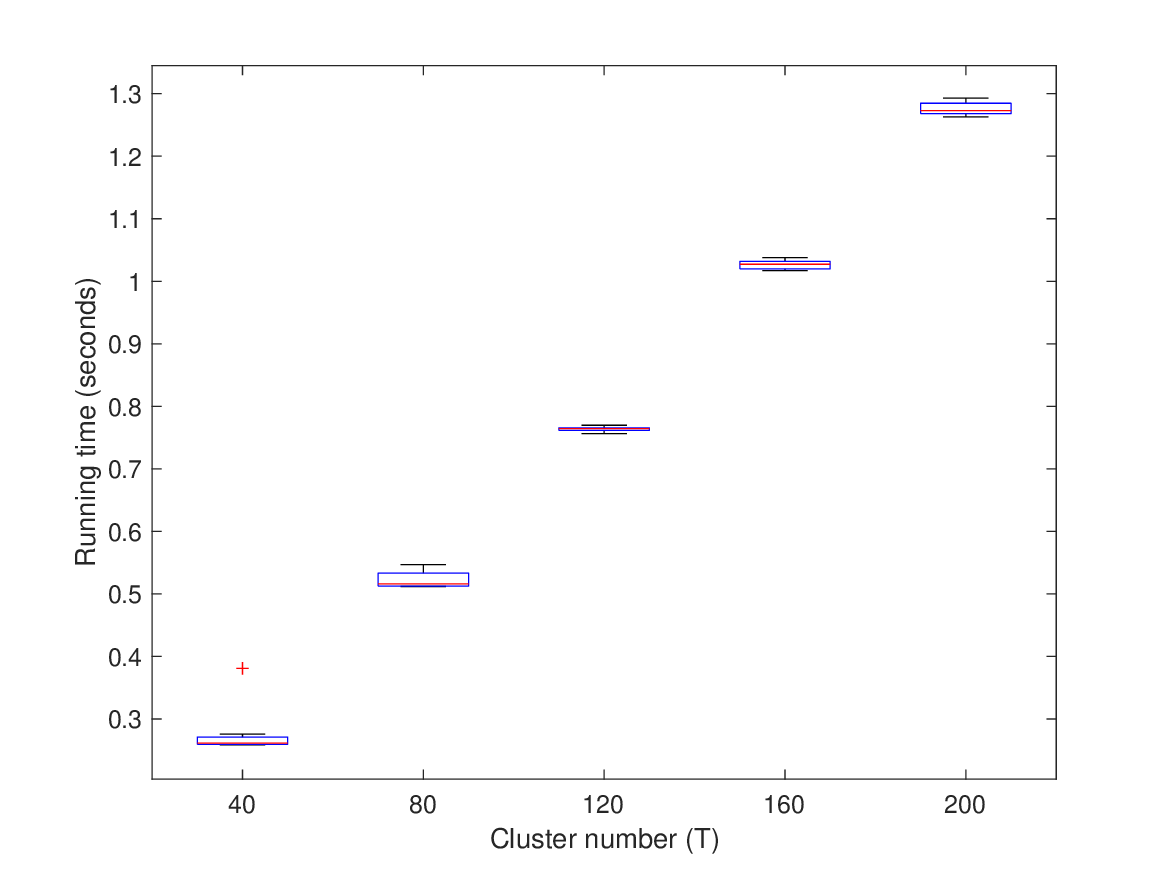}}
    \quad
    \subfloat[Reconstruction error for different cluster number.]{
	\includegraphics[width=0.4\textwidth]{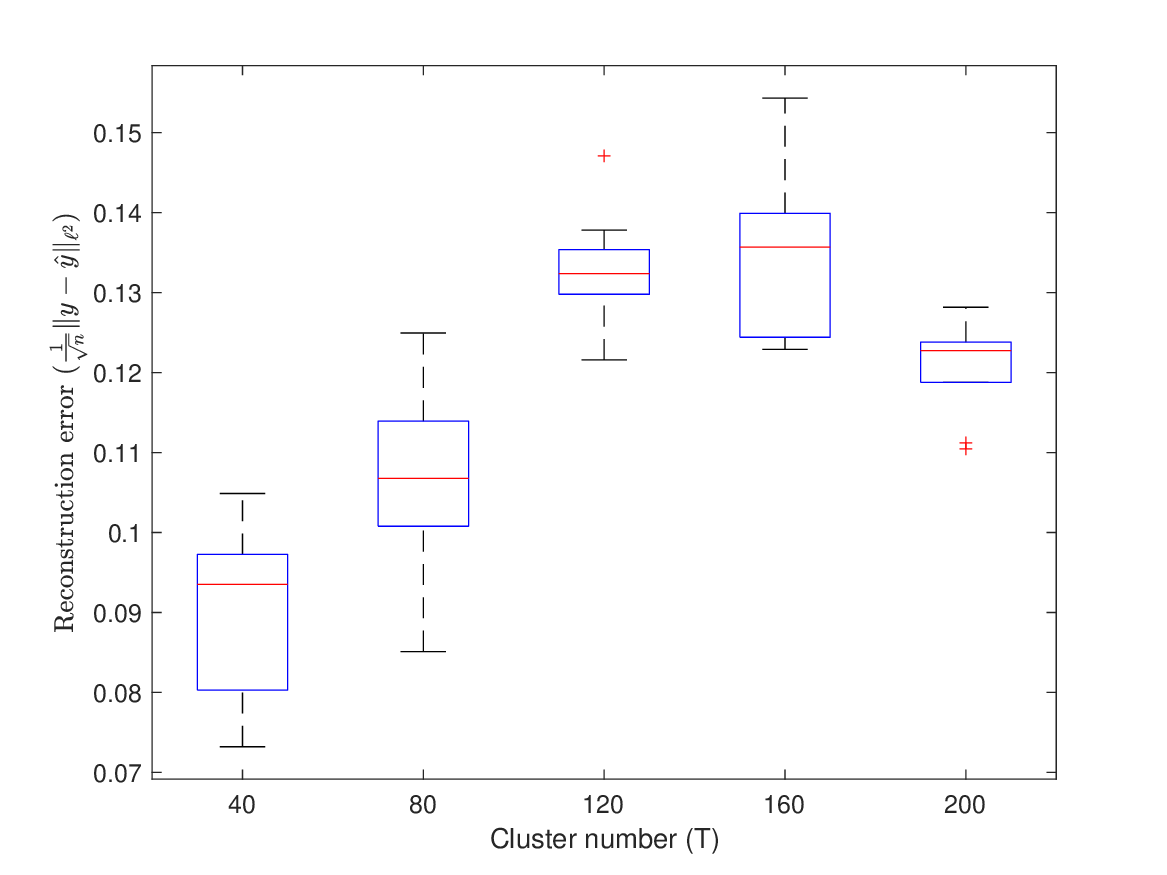}}
 \caption{Numerical result for cluster spectra reconstruction (2 spectra within a cluster)}
\label{Fig: random-time-2}
\end{figure}

For the second group, we set 3 line spectra in each cluster, the separation distance of the three line spectra is 1.2 and the distance between clusters is $7\pi$. We set $\Omega=1$, $\sigma = 10^{-4}$ and $h=0.001$. In both cluster center reconstruction and spectra reconstruction, we pick the Gaussian parameter, the trust level, and the subsampling factor as $(\lambda,\kappa_T, F_{sub})=(70,0.9,60)$. For the spectra reconstruction, we pick the multipole order for the clusters by the following strategy: we choose the multipole order to be 5 for the nearest cluster and 3 for others. The running time of SCAN-MUSIC(C) combined with SCAN-MUSIC and the reconstruction error for the reconstructed position is shown in Figure \ref{Fig: random-time-3}.
\begin{figure}[ht]
	\centering
    \subfloat[Running time of reconstruction for different cluster number.]{
	\includegraphics[width=0.4\textwidth]{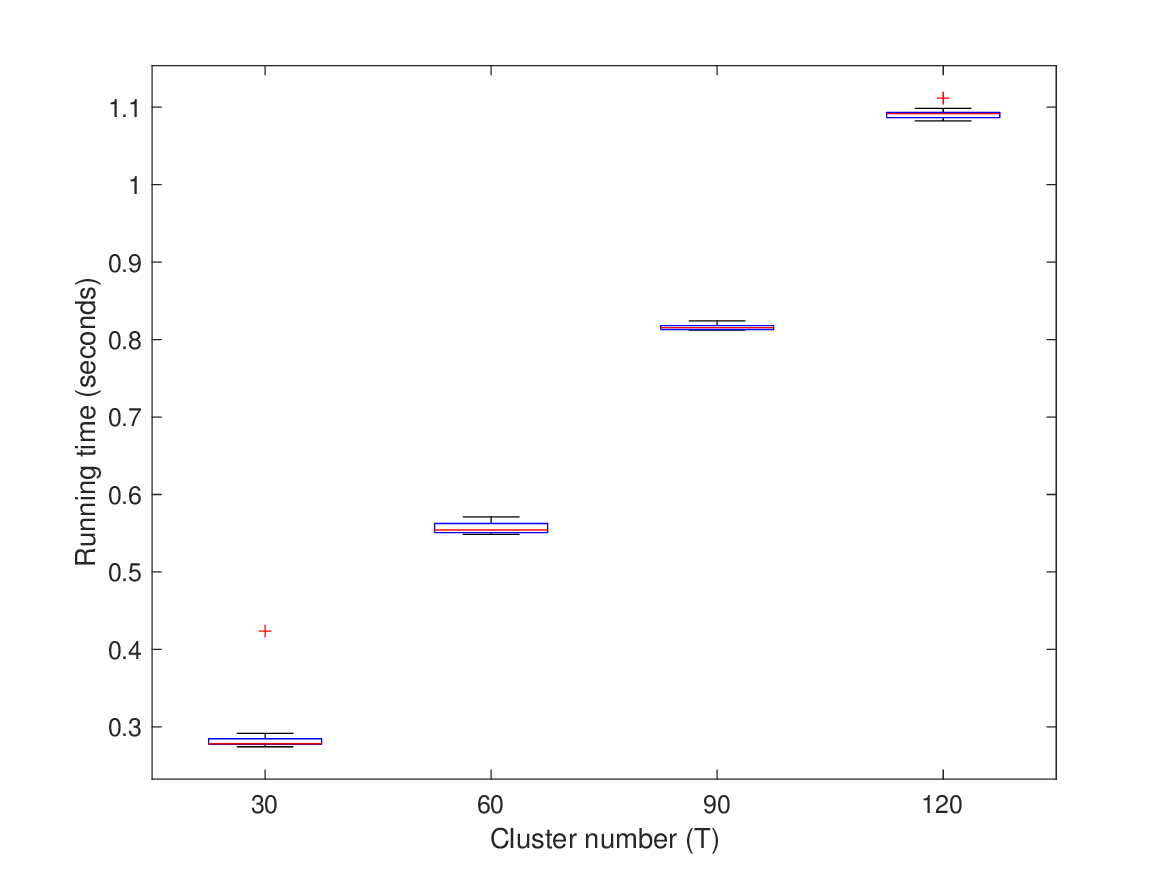}}
    \quad
    \subfloat[Reconstruction error for different cluster number.]{
	\includegraphics[width=0.4\textwidth]{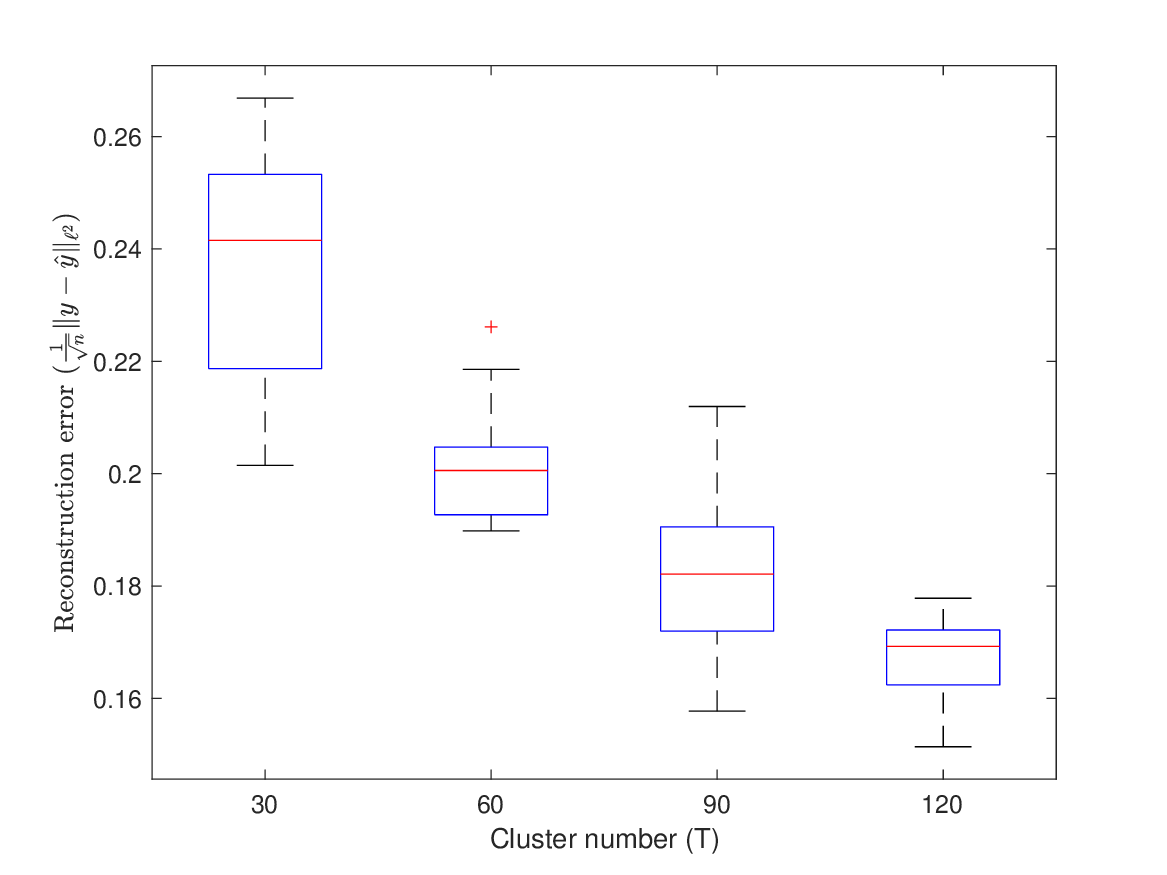}}
 \caption{Numerical result for cluster spectra reconstruction (3 spectra within a cluster)}
\label{Fig: random-time-3}
\end{figure}

\subsection{Efficiency Test}
In this section, we demonstrate the efficiency of SCAN-MUSIC by comparing it with the Superfast LSE method\footnote{The code for Superfast LSE is from https://github.com/thomaslundgaard/superfast-lse}  proposed in \cite{hansen2018superfast} that has demonstrated estimation accuracy at least as good as other current methods while being order of magnitudes faster. We consider both settings of random and clustered line spectra. \\

First, we set the line spectra to be randomly distributed on the spectral domain with minimum separation distance 2 Rayleigh limits. We fix the noise level to be $\mathcal{O}(10^{-2})$. For SCAN-MUSIC, we adopt the sampling step size to be $h=0.001$, and sample the interval $[-1,1]$. We choose the parameters $(\lambda,\kappa_T,F_{sub})=(70,0.9,60)$. For Superfast LSE, we adopt the same sample numbers. We conducted 20 random experiments for each setup. Figure \ref{Fig: random_lse_compare} shows the averaged running time for two different methods under different numbers of line spectra.\\

Second, we consider the cluster settings of line spectra. We set 2 line spectra with separation distance $1/\pi$ Rayleigh limit in each cluster. The distance between clusters is 6 Rayleigh limits. We fix the noise level to be $\mathcal{O}(10^{-3})$. For SCAN-MUSIC, we adopt the sampling step size to be $h=0.001$, and sample the interval $[-1,1]$. We choose the parameters $(\lambda,\kappa_T,F_{sub})=(100,0.95,60)$. For Superfast LSE, we adopt the same sample numbers. As reported in \cite{hansen2018superfast} when the separation distance is below the Rayleigh limit, Superfast LSE may not reconstruct the ground truth of all line spectra. This is indeed the case here. However, Scan-MUSIC succeeds in all reconstructions. Nevertheless, we pick the successful cases for both methods and report the averaged running time of the two methods in Figure \ref{Fig: cluster_lse_compare}.\\

\begin{figure}[ht]
	\centering
    \subfloat[Running time of reconstruction for random spectra.]{
    \label{Fig: random_lse_compare}
	\includegraphics[width=0.4\textwidth]{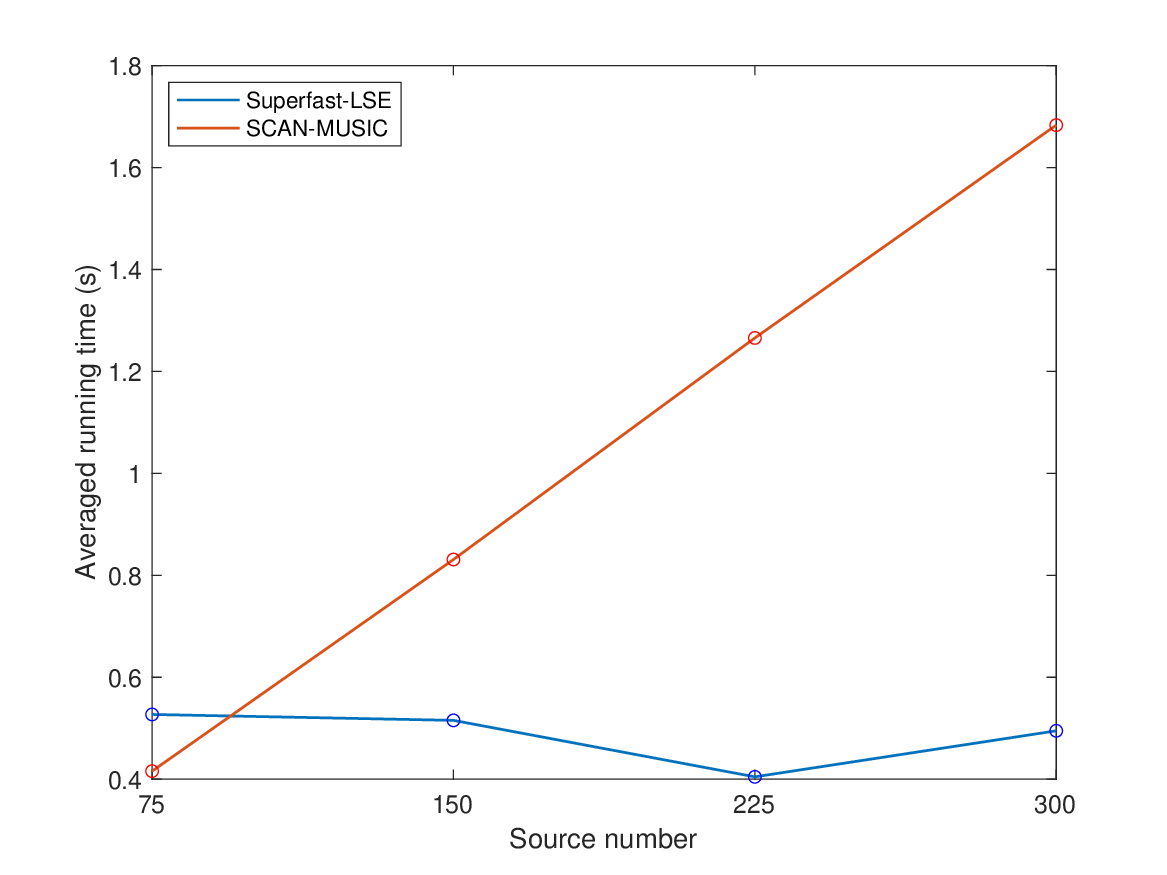}}
    \quad
    \subfloat[Running time of reconstruction for clustered spectra.]{
    \label{Fig: cluster_lse_compare}
	\includegraphics[width=0.4\textwidth]{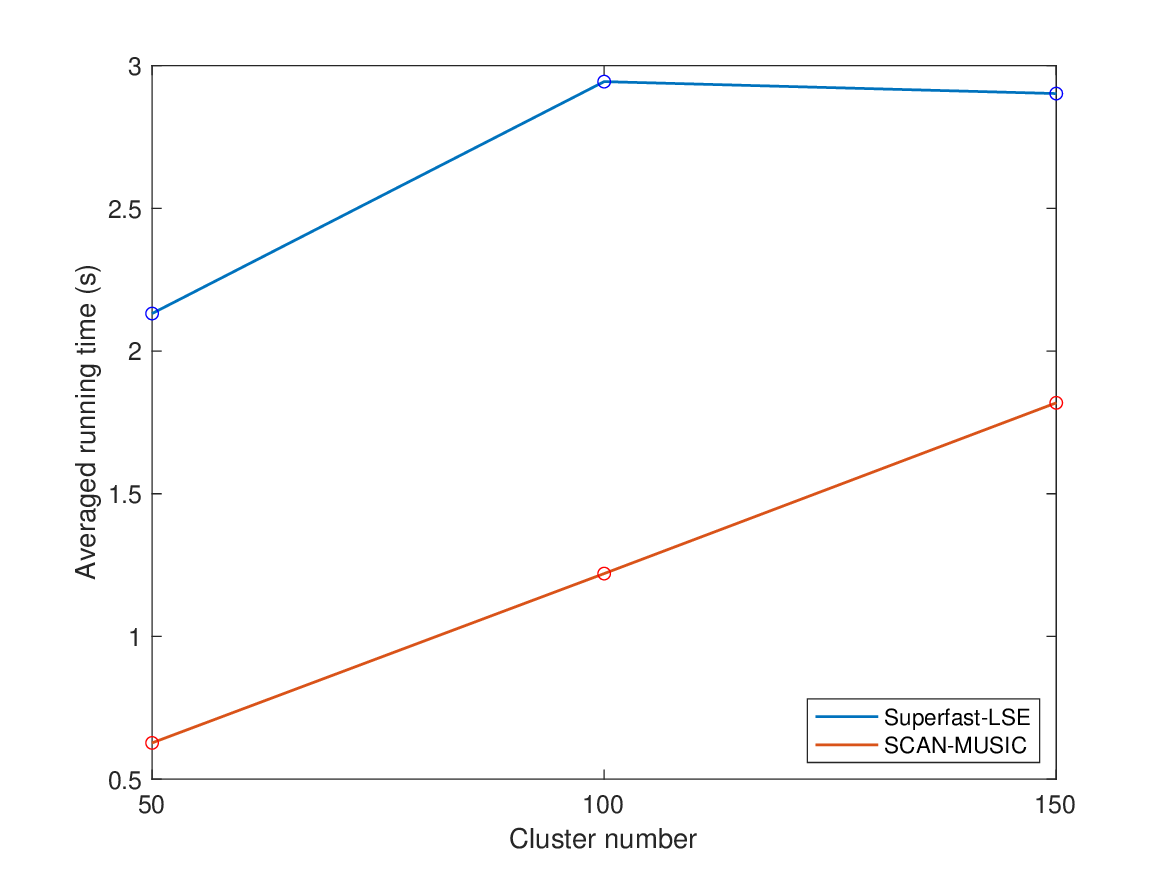}}
 \caption{Efficiency comparison for SCAN-MUSIC and Superfast LSE.}
\label{Fig: Efficiency compare}
\end{figure}

We observe that in the two groups of numerical experiments above, the averaged running time is linear to the source number. This is due to the fact that the sampling step size $h$ is fixed in all these numerical experiments and hence the sampling complexity $K$ is fixed as a constant. 
As a result, the computational complexity in both cases scales as $\mathcal{O} (n)$ according to formula (\ref{random-complex}) and (\ref{cluster-complex}) respectively.\\

Through the above numerical experiments, we observe that SCAN-MUSIC has an efficiency comparable to the state-of-the-art method, the Superfast LSE method. Moreover, it has unique strength in reconstructing line spectra with cluster structures where their separation distances may be below the Rayleigh limit. We also want to point out that SCAN-MUSIC can still be improved by applying computational techniques, like parallel computing for instance.

\section{Discussion} \label{sec: discussion}
In this paper, we propose efficient super-resolution algorithms SCAN-MUSIC and SCAN-MUSIC(C) for line spectral estimation. We present along with numerical experiments, the algorithm, error estimates, computational limit, sampling complexity, and computational complexity. SCAN-MUSIC adopts the Gaussian windowing, centralization, and MUSIC algorithm to achieve complete spectra reconstruction. For clustered line spectra, the annihilating filter technique is introduced and SCAN-MUSIC(C) is proposed. The two proposed algorithms can be easily paralleled.\\

We show that in the regime we consider, the two proposed algorithms achieve the optimal order of sampling complexity $\mathcal{O}(n)$ and computational complexity $\mathcal{O}(n^2\log n)$ (by appropriate choice of subsampling factor). We expect the two methods can be combined with strategies in multiple measurements to solve the LSE problem with multiple measurements and achieve reconstruction with higher resolution.

\section{Appendix}

\subsection{Error Estimates of SCAN-MUSIC}\label{subsec: error analysis}
To fix the notations in this section, we introduce the following notations.
We define $B = \max\{s\in\mathbb{N}:sh\le \Omega_{win} \}$, where $\Omega_{win}$ is defined in (\ref{def: omega_win}). We denote $\omega_{k} = kh$, for $k\in \mathbb{Z}$, where $h$ is the sampling step size defined in Section \ref{sec: introduction}.  We further denote
\begin{align*}
    &W_{cen} = (\mathcal{S}_\mu W (\omega_{-K}),\cdots,\mathcal{S}_\mu W (\omega_{K})), \quad
    f_{cen} = (\mathcal{S}_\mu f (\omega_{-K}),\mathcal{S}_\mu f (\omega_{-K+h}),\cdots,\mathcal{S}_\mu f (\omega_{K})),\\
    &
    \boldsymbol{G}_\lambda = h(\cdots,G_\lambda(\omega_{-1}),G_\lambda(\omega_{0}),G_\lambda(\omega_{1}),\cdots).
\end{align*}
For given $\eta\in\Real$, we define the function $\varphi_{\eta} (\zeta)$ by 
\begin{align}
    \varphi_{\eta} (\zeta) = \sum_{j=1}^n a_je^{i(y_j-\mu)\zeta}\cdot G_{\lambda}(\eta-\zeta),
\end{align}
where $\zeta \in \mathbb{C}$. Let $F_{\eta}(\xi)$ be the restriction of $F_{\eta} (\zeta)$ on the real line. It is easy to check that $F_{\eta} (\zeta)$ is analytic on the complex plane.\\
We notice that in the ideal case, the centralization and Gaussian windowing are characterized by the continuous convolution by the following identity:
\begin{align}
    \mathcal{S}_{\mu}(f) \ast G_{\lambda}
   = \sum_{j=1}^{n} a_j e^{-\frac{(y_j-\mu)^2}{4\lambda}}e^{i(y_j-\mu)\omega},
\end{align}
In practice, however, only the discrete band-limited noisy measurement is available, and the Gaussian window is also truncated by the level $\gamma$ as shown in (\ref{Y_win}).
The total error can be characterized by four factors: 1. the discretization error caused by the transfer of continuous convolution into discrete convolution, 2. the model error for the discrete case which is caused by the band-limited data, 3. the truncation error caused by the truncation of Gaussian window, 4. the noise error. We denote them by $\mathcal{E}_1$, $\mathcal{E}_2$, $\mathcal{E}_3$, $\mathcal{E}_4$ respectively. Mathematically, we characterize the above errors in the following 
\begin{defn}
For $\omega \in \Real$, $b\in\mathbb{Z}$ with $b\le B$, we define
    \begin{align} 
        &\mathcal{E}_1 :=\sup_{\omega\in\Real} |\mathcal{S}_{\mu}(f) \ast G_{\lambda}(\omega) - h\cdot \sum_{k\in\mathbb{Z}} F_{\omega}(\omega_{k})| \\
        &\mathcal{E}_2 := \max_{b = -B,\cdots,B} h\cdot \left|\sum_{k\in\mathbb{Z}} F_{\omega_b}(\omega_{k})-\sum_{k = -K}^K F_{\omega_b}(\omega_{k}) \right|\\
        &\mathcal{E}_3 := \| f_{cen}\ast(\boldsymbol{G}_\lambda-G_{\lambda,\Gamma}) \|_{\ell^\infty} \\
        &\mathcal{E}_4 := \|W_{cen}\ast G_{\lambda,\Gamma} \|_{\ell^\infty}
    \end{align}
\end{defn}
\begin{thm}
    Assume that the choice of $\gamma$ satisfies $\frac{\gamma}{\sqrt{-\ln \gamma}}\le\frac{\sqrt{\pi}}{\|\nu\|_{TV}\sigma}$. Then for any given $C>0$, we have 
    \begin{align}\label{e1+e2+e3+e4}
        \mathcal{E}_{total} :=
        \mathcal{E}_1+\mathcal{E}_2+\mathcal{E}_3+\mathcal{E}_4 
        \lesssim 
        \frac{2e^{\lambda C^2} \|\nu \|_{TV}}{e^{2\pi C/h}-1}+
        3\sigma.
    \end{align}
\end{thm}
\begin{proof}
We estimate the four errors one by one.\\

\textbf{Step 1.} For any $\omega \in \Real$, and any fixed $C >0$, it is clear that $\varphi_{\omega}(\zeta)$ is analytic in the strip $|\Im(\zeta)|<C$ and $\varphi_{\omega}(\zeta) \rightarrow 0$ uniformly as $|\zeta|\rightarrow \infty$ in the strip. For any given $c\in(-C,C)$, we have
\begin{align}
    \int_\Real |\varphi_{\omega}(\xi+ic)| d\xi
    & = \int_\Real \left|\sum_{j=1}^n a_je^{i(y_j-\mu)\xi}\cdot G_{\lambda}(\omega-\xi-ic)\right| d\xi \notag \\
    & \le \|\nu \|_{TV} \cdot \sqrt{\frac{\lambda}{\pi}}\int_\Real |e^{-\lambda(\omega-\xi-ic)^2}| d\xi \notag \\
    & = \|\nu \|_{TV}\cdot e^{\lambda c^2} < e^{\lambda C^2} \cdot \|\nu \|_{TV}.
\end{align}
By Theorem 5.1 in \cite{doi:10.1137/130932132}, we have
\begin{align}
    \mathcal{E}_1 \le \frac{2e^{\lambda C^2} \|\nu \|_{TV}}{e^{2\pi C/h}-1}.
\end{align}
\textbf{Step 2.} For any given $b = -B,\cdots,B$, we have 
\begin{align}
    \left | \sum_{k = -\infty}^{-K-1} \left( \sum_{j=1}^n a_je^{i(y_j-\mu)\omega_k}\cdot hG_{\lambda}(bh-kh) \right) \right|
    & \lesssim \|\nu \|_{TV} \cdot \sqrt{\frac{\lambda}{\pi}}\int_{-\infty}^{-\Omega} e^{-\lambda(bh-\omega)^2} d\omega\notag\\
    & = \frac{\|\nu \|_{TV}}{\sqrt{\pi}}\cdot \Phi\left(\sqrt{\lambda}(-\Omega-bh)\right).
\end{align}
Similarly, we have 
\begin{align}
    \left | \sum_{k = K+1}^{+\infty} \left( \sum_{j=1}^n a_je^{i(y_j-\mu)\omega_k}\cdot hG_{\lambda}(bh-kh) \right) \right|
    \lesssim \frac{\|\nu \|_{TV}}{\sqrt{\pi}}\cdot \Phi\left(\sqrt{\lambda}(bh-\Omega)\right).
\end{align}
Combining the two inequalities above, we have
\begin{align}
    \mathcal{E}_2 
    &\lesssim \frac{\|\nu \|_{TV}}{\sqrt{\pi}} \left(\Phi\left(\sqrt{\lambda}(-\Omega-Bh)\right)+\Phi\left(\sqrt{\lambda}(Bh-\Omega)\right) \right)\le \sigma,
\end{align}
where the last inequality is due to the definition of $B$ and $\Omega_{win}$.\\
\textbf{Step 3.} 
Straightforward calculation gives
\begin{align}
    \mathcal{E}_3 \le \|W_{cen}\|_{\ell^\infty} \cdot \|G_{\lambda,\Gamma}\|_{\ell^1} \lesssim \sigma\cdot  \sqrt{\frac{\lambda}{\pi}}\int_{-\Gamma h}^{\Gamma h} e^{-\lambda \omega^2} d\omega < \sigma.
\end{align}
\textbf{Step 4.} 
Combining the definition of $\Gamma$, see (\ref{def:gamma}), and the assumption on $\gamma$, we have
\begin{align}
    \mathcal{E}_4
    &\le \|Y_{cen}\|_{\ell^\infty} \cdot \|(\boldsymbol{G}_{\lambda}-G_{\lambda,\Gamma})\|_{\ell^1}\notag \\
    &\lesssim 2\| \nu \|_{TV} \cdot \sqrt{\frac{\lambda}{\pi}}\int_{\Gamma h}^{+\infty} e^{-\lambda \omega^2}d\omega \notag\\
    &\le 2\| \nu \|_{TV} \cdot \sqrt{\frac{\lambda}{\pi}}\int_{\Gamma h}^{+\infty}  \frac{\omega}{\Gamma h}\cdot e^{-\lambda \omega^2}d\omega \notag \\
    & =\frac{\| \nu \|_{TV}}{\sqrt{\lambda\pi}} \cdot \frac{e^{-\lambda(\Gamma h)^2}}{\Gamma h} \le \frac{\| \nu \|_{TV}}{\sqrt{\pi}}\cdot\frac{\gamma}{\sqrt{-\ln \gamma}}\le \sigma
\end{align}
Combining the discussion above, we get the estimation as in (\ref{e1+e2+e3+e4}).
\end{proof}
\begin{rmk}
    In practice, we pick the Gaussian parameter $\lambda \sim 10^2$. Notice that the sampling step size $h$ is small due to the Nyquist-Shannon sampling for wide-band signal. Then, for $C = 0.1$, one can see that $\mathcal{E}_1$ is sufficiently small and we can conclude that $\mathcal{E}_{total}\lesssim \sigma$.
\end{rmk}
\begin{rmk}
    As pointed in Section \ref{subsec: scan-music}, in practice, the choice $\gamma\sim\mathcal{O}(\sigma)$ can result in stable reconstruction.
\end{rmk}

\subsection{Auxiliary Lemmas}
\begin{lem} \label{lem: gauss-conv}
    For fixed $\xi \in \Real$, we have
    \begin{align}
        e^{i\xi\omega} \ast G_{\lambda}(\omega) = e^{-\frac{\xi^2}{4\lambda}}\cdot e^{i\xi\omega}
    \end{align}
\end{lem}
\begin{proof}By direct calculation, we have
\begin{align}
    e^{i\xi\omega} \ast G_{\lambda}(\omega) 
    &= \mathcal{F} \left(\mathcal{F}\left( e^{i\xi\omega} \ast G_{\lambda}(\omega)  \right) \right)\notag\\
    &= \mathcal{F}( 2\pi \delta_{\xi}\cdot e^{-\frac{x^2}{4\lambda}} )\notag\\
    &=e^{-\frac{\xi^2}{4\lambda}}\cdot e^{i\xi\omega}
\end{align}
\end{proof}

\begin{lem}\label{lem: gauss-ineq}
    For $x \in \Real$, we have
    \begin{align}
        \Phi (-x) < \frac{e^{-x^2}}{2x}.
    \end{align}
\end{lem}
\begin{proof}
Since $\Phi(x)$ is an even function, we have
\begin{align} 
    \Phi (-x) = \int_{-\infty}^{-x} e^{-t^2} dt = \int_{x}^{\infty} e^{-t^2} dt 
    < \int_{x}^{\infty}\frac{t}{x} \cdot e^{-t^2} = \frac{e^{-x^2}}{2x}.
\end{align}
    
\end{proof}

\subsection{Auxiliary Algorithms}
We present the auxiliary algorithms used in SCAN-MUSIC and SCAN-MUSIC(C) in this section.
\begin{algorithm}
	\caption{Centralization and Gaussian windowing ($CGM$)}
	\label{algo_gauss}
	\DontPrintSemicolon
    \SetKwInOut{Input}{Input}
    \SetKwInOut{Output}{Output}
    \SetKwInOut{Init}{Initialization}
	\Input{Original measurement: $Y$, sample interval: $[-\Omega,\Omega]$, sampling step size: $h$.}
    \Input{Gaussian parameter: $\lambda$, Gaussian center: $\mu$, truncation: $\gamma$}
		$grid \gets -\Omega:h:\Omega$ \;
		$G_{pre} \gets \exp(-\lambda\cdot grid.^{\wedge} 2)$\;
		$G \gets G_{pre}(G_{pre}>\gamma)$\;
		$Y_{center} \gets Y .\ast \exp(-i\mu\cdot grid)$\;
        $Y_{win} \gets conv(Y_{center},G,'\text{valid}')$\;
	\Return $Y_{win}$.\;
\end{algorithm}

\begin{algorithm}
	\caption{Subsampling for random line spectra ($Sub_1$)}
	\label{algo_sub1}
	\DontPrintSemicolon
    \SetKwInOut{Input}{Input}
    \SetKwInOut{Output}{Output}
    \SetKwInOut{Init}{Initialization}
	\Input{Modulated measurement: $Y_{win}$, bandwidth after windowing $\Omega_{\gamma}$, sampling step size: $h$.}
    \Input{subsampling factor: $F_{sub}$}
		$index_{pre} \gets 0:[\frac{|Y_{win}|}{F_{sub}}]-1 $\;
        $index \gets F_{sub}\cdot index_{pre}+1$\;
		$Y_{sub} \gets Y_{win}(index)$\;
		$h_{sub} \gets h\cdot F_{sub}$\;
	\Return $Y_{sub}$, $h_{sub}$.\;
\end{algorithm}

\begin{algorithm}
	\caption{Subsampling for Clustered line spectra ($Sub_2$)}
	\label{algo_subsample_cluster}
	\DontPrintSemicolon
    \SetKwInOut{Input}{Input}
    \SetKwInOut{Output}{Output}
    \SetKwInOut{Init}{Initialization}
	\Input{Modulated measurement: $Y_{win}$, Cluster center to be filtered: $y_{center}$, sampling step size: $h$}
    \Input{Corresponding pole order: $Ord$, spectra number in single cluster: $N_0$} 
    \tcc{$Ord$ should be a vector with the same length as $y_{center}$.}
        $M = sum(Ord)+1$\;
        $F_{sub} \gets \lceil\frac{|Y_{win}|}{2*N_0+M}\rceil$ \;
        $index \gets 1:F_{sub}:|Y_{win}| $ \;
        $Y_{sub} \gets Y_{win}(index)$, $h_{sub} \gets h\cdot F_{sub}$ \;
	\Return $Y_{sub}$, $h_{sub}$\;
\end{algorithm}

\begin{algorithm}
	\caption{Annihilating Filter Based spectra Removal ($AFSR$)}
	\label{algo_filter}
	\DontPrintSemicolon
    \SetKwInOut{Input}{Input}
    \SetKwInOut{Output}{Output}
    \SetKwInOut{Init}{Initialization}
	\Input{Subsampling measurement: $Y_{sub}$, Cluster center to be filtered: $y_{center}$, sampling step size: $h$} 
    \Input{Corresponding pole order: $Ord$} 
        \For{$t\ =\ 1:|y_{center}|$}{
            Construct $F_t$ based on (\ref{filter frequency domain}) 
            }
        $F \gets F_1 \ast \cdots \ast F_T$ \;
        $Y_{filt} \gets conv(Y_{sub},F,'valid')$ \; 
	\Return $Y_{filt}$\;
\end{algorithm}

\newpage
\subsection{Notation for Parameters}\label{subsec: notations}
\begin{table}[ht]
    \centering
    \begin{tabular}{cc}
   \toprule
   Parameters & Notation \\
   \midrule
   Effective cutoff frequency & $\Omega_{win}$\\
   Gaussian Parameter & $\lambda$ \\
   Gaussian Center & $\mu$ \\
   Truncation level for Gaussian window & $\gamma$ \\
   Truncation index for Gaussian window &  $\Gamma$ \\
   Trust level & $\kappa_T$\\
   Trust bound & $R_{tru}$\\
   Essential level & $\kappa_E$\\
   Essential bound & $R_{ess}$\\
   Subsampling factor & $F_{sub}$\\
   Measurement after centralization and Gaussian windowing & $Y_{win}$\\
   $Y_{win}$ further processed by subsampling & $Y_{sub}$\\
   $Y_{sub}$ further processed by filtering & $Y_{filt}$\\
   \bottomrule
\end{tabular}
    \caption{Table for parameters and notations}
    \label{Tab: notations}
\end{table}

\newpage
\bibliographystyle{ieeetr}
\bibliography{references} 
\newpage

\end{document}